\documentclass[11pt]{article}
\usepackage{amssymb}
\usepackage{amsthm}
\usepackage{amsmath}
\usepackage{graphicx}
\usepackage{lineno}
\usepackage{hyperref}
\usepackage{fullpage}
\usepackage{microtype}
\usepackage{tabularx}
\usepackage{booktabs}

\theoremstyle{plain}
\newtheorem{theorem}{Theorem}[section]
\newtheorem{lemma}[theorem]{Lemma}

\newtheorem{claim}[theorem]{Claim}
\newtheorem{corollary}[theorem]{Corollary}

\newtheorem{fact}[theorem]{Fact}

\providecommand{\abs}[1]{\lvert#1\rvert} \providecommand{\norm}[1]{\lVert#1\rVert}
\providecommand{\mat}[1]{\ensuremath{\mathbf{#1}}}

\newcommand{\dmax}{d_{\mathsf{max}}}

\newcommand{\deltamin}{\delta_{\mathsf{min}}}

\newcommand{\wmin}{w_{\mathsf{min}}}
\newcommand{\myin}{\mathsf{in}}
\newcommand{\myout}{\mathsf{out}}
\newcommand{\row}{\mathrm{row}}
\newcommand{\ball}{\mathsf{ball}}
\newcommand{\vol}{\mathrm{vol}}
\newcommand{\bad}{\mathrm{bad}}
\newcommand{\etal}{\emph{et al.}}
\newcommand{\y}{\mathbf{y}}
\newcommand{\vv}{\mathbf{v}}
\newcommand{\vvv}{\mathbf{v'}}
\newcommand{\f}{\mathbf{F}}
\newcommand{\xx}{\mathbf{x}}
\newcommand{\pp}{\mathbf{p}}
\newcommand{\LL}{\mathbf{\cal L}}
\newcommand{\cst}{{c}}

\renewcommand{\nu}{\ensuremath{\mathbf{g}}}

\def\polylog{\operatorname{polylog}}

\DeclareMathOperator*{\argmax}{argmax}
\DeclareMathOperator*{\mytriangle}{\triangle}

\newcommand{\eps}{\varepsilon}
\title{Spectral Concentration and Greedy $k$-Clustering}
\author{
Tamal K. Dey\thanks{
    Dept.~of Computer Science and Engineering, and Dept.~of Mathematics, The Ohio State University. Columbus, OH, 43201.
    \texttt{tamaldey@cse.ohio-state.edu}
}
\and
Pan Peng\thanks{
    Department of Computer Science, University of Sheffield.
    \texttt{p.peng@sheffield.ac.uk}
}
\and
Alfred Rossi\thanks{
    Dept.~of Computer Science and Engineering, The Ohio State University. Columbus, OH, 43201.
    \texttt{rossi.49@osu.edu}
}
\and
Anastasios Sidiropoulos\thanks{
    Dept.~of Computer Science and Engineering, and Dept.~of Mathematics, The Ohio State University. Columbus, OH, 43201.
    \texttt{sidiropoulos.1@osu.edu}
}
}
\date{}
\begin{document}
\maketitle

\begin{abstract}
A popular graph clustering method is to consider the embedding of an input graph
 into $\mathbb{R}^k$ induced by the first $k$ eigenvectors of its Laplacian, and
to partition the graph via geometric manipulations on the resulting metric
space. Despite the practical success of this methodology, there is limited
understanding of several heuristics that follow this framework. We provide
theoretical justification for one such natural and computationally efficient
variant.

Our result can be summarized as follows. A partition of a graph is called {\em
strong} if each cluster has small external conductance, and large internal
conductance. We present a simple greedy spectral clustering algorithm which
returns a partition that is provably close to a suitably strong partition,
provided that such a partition exists. A recent result shows that strong
partitions exist for graphs with a sufficiently large spectral gap between the
$k$-th and $(k+1)$-st eigenvalues. Taking this together with our main theorem
gives a spectral algorithm which finds a partition close to a strong one for
graphs with large enough spectral gap. We also show how this simple greedy
algorithm can be implemented in near-linear time for any fixed $k$ and error
guarantee. Finally, we evaluate our algorithm on some real-world and synthetic
inputs.
\end{abstract}

%\linenumbers
%\maketitle
    
 \section{Introduction}
Spectral clustering of graphs is a fundamental technique in data analysis that
has enjoyed broad practical usage because of its efficacy and simplicity. The
technique maps the vertex set of a graph into a Euclidean space $\mathbb{R}^k$
 where a classical clustering algorithm (such as $k$-means, $k$-center) is
applied to the resulting embedding~\cite{DBLP:journals/sac/Luxburg07}. The
coordinates of the vertices in the embedding are computed from $k$ eigenvectors
of a matrix associated with the graph. The exact choice of matrix depends on the
specific application but is typically some weighted variant of $D-A$, for a
 graph with degree matrix $D$ and adjacency matrix $A$. Despite widespread usage, theoretical
understanding of the technique remains limited. For example, it is generally not
clear for which classes of graphs spectral clustering works well, or what the
structure of the subgraph induced by vertices that correspond to embedded points
from the same cluster is. Although the case for $k=2$ (two clusters) is well
understood, the case of general $k$ is not yet settled and a growing body of
work seeks to address the practical success of spectral clustering methods
\cite{DBLP:conf/nips/BalakrishnanXKS11,DBLP:journals/jacm/BiswalLR10,
DBLP:journals/siamcomp/Kelner06, DBLP:conf/nips/NgJW01,
DBLP:conf/compgeom/SpielmanT96, DBLP:journals/sac/Luxburg07}.
 
In this paper we present a simple greedy spectral clustering algorithm which is
guaranteed to return a high quality partition, provided that one of sufficient
quality {\em exists}. It first chooses $k$ clusters along with their centers
greedily from the vertices spectrally embedded in an Euclidean space. Any left
over vertex is assigned to one of the computed clusters whose center it is
closest to. The resulting partition is close in symmetric difference to the high
quality one. Our results can be viewed as providing further theoretical
justification for popular clustering algorithms such as in
\cite{DBLP:conf/nips/BalakrishnanXKS11} and \cite{DBLP:conf/nips/NgJW01}.

\paragraph*{Measuring partition quality} Intuitively, a good $k$-clustering
of a graph is one where there are few edges between vertices residing in
different clusters and where each cluster is well-connected as an induced
subgraph. Such a qualitative definition of clusters can be appropriately
characterized by vertex sets with small external conductance and large internal
conductance, which has been first formalized by Oveis Gharan and
Trevisan~\cite{DBLP:conf/soda/GharanT14}.

Let $G=(V,E)$ be an undirected unweighted graph. %with $w(u,v) > 0$ if $(u,v)\in E$ and $w(u,v)=0$, otherwise. 
Let $\deg(v)$ % :=\sum_{u\in V}w(v,u)$ 
be the degree of a vertex $v\in V$. For a subset $S \subset V$, the \emph{external
conductance} and \emph{internal conductance} are defined to be
\[
\phi_{\myout}(S; G) := \frac{|E(S, V(G) \setminus S)|}{\vol (S)}
\text{,~~~~}
\phi_{\myin}(S) := \min_{S'\subseteq S, \vol(S')\leq \frac{\vol(S)}{2}}
\phi_{\myout}(S'; G[S])
\]
respectively, where $\vol(S) = \sum_{v\in S}\deg(v)$ (called the \emph{volume}
of $S$), $E(X,Y)$ denotes the set of edges between $X$ and $Y$, and
$G[S]$ denotes the subgraph of $G$ induced on $S$. For an {\em isolated} vertex
$v$ in $G$, we assume $\phi_{\myout}(v,G)=0$ and $\phi_{\myin}(v)=1$ by
definition. Let $\phi_\myin(G):=\phi_\myin(V)$. It follows that if
$\phi_{\myin}(G)>0$, then $G$ cannot have any isolated vertex. When $G$ is
understood from context we sometimes write $\phi_{\myout}(S)$ in place of
$\phi_{\myout}(S; G)$.

We define a \emph{$k$-partition} of a graph $G$ to be a partition ${\cal
A}=\{A_1,\ldots,A_k\}$ of $V(G)$ into $k$ disjoint subsets. We say that ${\cal
A}$ is \emph{$(\alpha_{\myin}, \alpha_{\myout})$-strong}, for some
$\alpha_{\myin}, \alpha_{\myout}\geq 0$, if for all $i\in \{1,\ldots,k\}$, we
have
\[
\phi_{\myin}(A_i) \geq \alpha_{\myin}
\text{ ~~~~ and ~~~~ }
\phi_{\myout}(A_i) \leq \alpha_{\myout}.
\]

Thus a high quality partition is one where $\alpha_{\myin}$ is large and
$\alpha_{\myout}$ is small.

\paragraph*{Our contribution} We present a simple spectral algorithm which
computes a partition provably close to {\em any} $(\alpha_{\myin},
\alpha_{\myout})$-strong $k$-partition if there is large gap between
$\alpha_{\myin}$ and $\alpha_\myout$ (see Theorem~\ref{thm:main} for formal
statement). We emphasize the fact that the algorithm's output approximates any
good existing clustering in the input graph. The algorithm consists of a simple
greedy clustering procedure performed on the embedding into $\mathbb{R}^k$
induced by the first $k$ eigenvectors. We further show how to implement this
algorithm in near-linear time for any fixed $k$ and error guarantee (see
Theorem~\ref{thm:algorithm2}).

In the analysis of our algorithm, we show some interesting spectral properties
of graphs that admit strong $k$-partitions: each of the (rescaled) first $k$
eigenvectors of the Laplacian matrix of the graph is close to some vector that
is constant on each cluster; the image of each cluster concentrates around some
point in the spectral embedding, and all these points are well separated.

\paragraph*{Related work} The discrete version of Cheeger's inequality
asserts that a graph admits a bipartition into two sets of small external
conductance if and only if the second eigenvalue is small
\cite{DBLP:journals/combinatorica/Alon86, DBLP:journals/jct/AlonM85, cheeger,
DBLP:conf/focs/Mihail89, DBLP:journals/iandc/SinclairJ89}. In fact, such a
bipartition can be efficiently computed via a simple algorithm that examines the
second eigenvector. Generalizations of Cheeger's inequality have been obtained
by Lee, Oveis Gharan, and Trevisan \cite{DBLP:conf/stoc/LeeGT12}, and Louis
\etal~\cite{DBLP:conf/stoc/LouisRTV12}. They showed that spectral algorithms can
be used to find $k$ disjoint subsets, each with small external conductance,
provided that the $k$-th eigenvalue is small. An improved version of Cheeger's
inequality has been obtained by Kwok \etal~\cite{DBLP:conf/stoc/KwokLLGT13} for
graphs with large $k$-th eigenvalue.

Even though the clusters given by the above spectral partitioning methods have
small external conductance, they are not guaranteed to have large internal
conductance. In other words, for a resulting cluster $C$, the induced graph
$G[C]$ might admit further partitioning into sub-clusters of small conductance.
Kannan, Vempala and Vetta proposed quantifying the quality of a partition by
measuring the internal conductance of clusters
\cite{DBLP:journals/jacm/KannanVV04}. Allen Zhu, Lattanzi and
Mirrokni~\cite{DBLP:conf/icml/ZhuLM13} and Orecchia and Allen Zhu~\cite{DBLP:conf/soda/OrecchiaZ14}
studied local algorithms for extracting subsets with small external conductance
under the assumption that subsets with small external conductance and high
(internal) connectivity exist.

One may wonder under what conditions a graph admits a partition which provides
guarantees on both internal and external conductance. Oveis Gharan and Trevisan,
improving on a result of Tanaka~\cite{Tanaka:2011vk}, showed that graphs which
have a sufficiently large spectral gap between the $k$-th and $(k+1)$-st
eigenvalues (denoted as $\lambda_k$ and $\lambda_{k+1}$, respectively) of its
Laplacian admit a strong $k$-partition~\cite{DBLP:conf/soda/GharanT14} (see
Theorem~\ref{thm:OT_clustering}). Czumaj et al.~\cite{DBLP:conf/stoc/CzumajPS15} recently
proposed a sublinear algorithm for testing if a graph with bounded maximum
degree has an $(\alpha_{\myin},\alpha_{\myout})$-strong partition in the
framework of property testing, assuming there is some gap between
$\alpha_{\myin},\alpha_{\myout}$.

\paragraph*{Follow-up work} Subsequent to the original ArXiv submission
\cite{DBLP:journals/corr/DeyRS14} of this paper, Peng, Sun, and Zanetti
\cite{DBLP:conf/colt/PengSZ15}, Awasthi et al.~\cite{DBLP:conf/innovations/AwasthiCKS16}
and  Sinop~\cite{DBLP:conf/soda/Sinop16} have derived spectral algorithms with weaker
assumption on the gap between $\alpha_{\myin}$ and $\alpha_{\myout}$ (or some
related gap, e.g., $\lambda_{k+1}$ and $\alpha_{\myout}$) to cluster the
vertices of the graph. The clustering algorithm analyzed in this paper remains
distinct from this body of work. For example, in
\cite{DBLP:conf/colt/PengSZ15} the authors applied $k$-means clustering to
the spectral embedding by the first $k$ eigenvectors; and show that the
resulting algorithm is able to find $k$ sets each of which is close to one
cluster of a strong $k$-partition and has bounded small external conductance,
under the assumption that $\lambda_{k+1}/\alpha_{\myout}=\Omega(k^3)$. (In
contrast, our assumption is $\alpha_{\myin}^2/\lambda_k=\Omega(k^3\dmax^3)$ for
graphs with maximum degree at most $\dmax$; see Theorem~\ref{thm:main}.) Their
error guarantee ultimately depends on the approximation factor afforded by a
$k$-means algorithm. Unfortunately, $k$-means is very sensitive to the initial
choice of $k$ centers and it is NP-hard to approximate to within some constant
factor~\cite{DBLP:conf/compgeom/AwasthiCKS15}. They also gave a heat-kernel
based algorithm that runs in near-linear time, which seems to be unappealing for
implementation. In \cite{DBLP:conf/innovations/AwasthiCKS16}, the authors proposed an algorithm
that iteratively applies the $k$-means clustering on the resistive embedding
projected onto the first $k$ eigenvectors, and outputs a $k$-partition such that
each part is close to one set in a target partition, under the assumption that
the ratio between the algebraic expansion of the clusters and $\alpha_{\myout}$
is $\Omega(k)$.   Sinop~\cite{DBLP:conf/soda/Sinop16} gave another spectral algorithm assuming
that $\lambda_{k+1}/\alpha_{\myout}=\Omega(1)$. In each case,
\cite{DBLP:conf/innovations/AwasthiCKS16,DBLP:conf/colt/PengSZ15,DBLP:conf/soda/Sinop16} use either a
different measure of the difference of the output partition and the target
partition, or a different definition of conductance (see our remark below
Theorem~\ref{thm:main}).

\paragraph*{Outline} Section~\ref{sec:greedy-k-clustering} contains a
description of the greedy $k$-clustering algorithm and the statement of our main
theorem. In Section~\ref{sec:concentration} we show a spectral concentration
property for any graph that admits a high quality partition. Building on this
property, we argue that the image of each cluster concentrates around some point
in the spectral embedding and these points are well separated. The complete proof of the main theorem is then given in Section~\ref{sec:clustering}.
In Section~\ref{fast-algo}, we give a randomized version of the algorithm which runs
in time $\tilde{O}(mk+\eps^{-1} k^3 n)$ for any error parameter $\eps>0$.
Finally, we present some experimental results in \ref{sec:experiments}.

\section{Greedy \texorpdfstring{$k$}{k}-Clustering}\label{sec:greedy-k-clustering} Let $G$ be an
undirected unweighted graph with $n$ vertices, and let $\mat{\LL}_G = \mat{I} -
\mat{D}^{-1/2} \mat{A} \mat{D}^{-1/2}$ be its normalized Laplacian, where
$\mat{A}$ is the adjacency matrix of $G$ and $\mat{D}$ is a diagonal
matrix with the entries $\mat{D}(v,v)$ equal to the degree of vertex
$v$. Let $0=\lambda_1\leq \lambda_2\leq \ldots \leq \lambda_n$ be the
eigenvalues of $\LL_G$, and $\vv_1,\vv_2 \ldots,\vv_n \in \mathbb{R}^n$ a
corresponding collection of orthonormal left eigenvectors\footnote{We will use
$\xx$ to denote a row vector and $\xx^T$ to denote a column vector.}. Note that
by the variational characterization of eigenvalues,
$\frac{\vv_i\LL_G\vv_i^T}{\vv_i\vv_i^T}=\lambda_i$ for $1\leq i\leq n$
(see~\cite{Chu97:spectral}).

In this paper we consider a simple geometric clustering operation on the
embedding $\f(u)$ which carries a vertex $u$ to a point given by a rescaling of
the first $k$ eigenvectors of $\LL_G$,
\begin{equation}\label{eq:embedding}
  \f(u) =  \deg(u)^{-1/2} \left(\vv_1(u), \ldots, \vv_k(u)\right).
\end{equation}
%
%where $\deg_G(u)$ denotes the degree of $u$ in $G$.  
For any $U
\subseteq V(G)$, let $\f(U)$ denote all the embedded points corresponding to
vertices in $U$, that is, $\f(U) = \{\f(u) : u \in U\}$. For any set $B\subseteq
\mathbb{R}^k$, let $\f^{-1}(B): = \{v\in V(G): \f(v)\in B\}$. For any point $\xx
\in \mathbb{R}^k$ and real number $R\geq 0$, let $\ball(\xx,R):=\{\y:
\norm{\y-\xx}_2 \leq R \}$.

\paragraph{Intuitive description}
The algorithm takes as input a graph $G$, and a desired number of clusters, $k$.
The algorithm uses the bottom $k$ eigenvectors $\vv_1, \ldots, \vv_k$ of $\LL_G$ to compute the embedding $\mathcal{F} = \f(V(G))$ of $G$ into $\mathbb{R}^k$.
Next, it begins an iterative process of searching for regions of the embedding containing many points from $\mathcal{F}$, and removing them to form clusters. To do so, it first computes a distance threshold $R = R(k, G) = R(k, n, \dmax)$, where $n = |V(G)|$ and $\dmax$ represent an upper bound of the maximum degree.
Using this threshold, the algorithm looks for a point $p \in \mathcal{F}$ such that the number of near-by points of $\mathcal{F}$ (points of $\mathcal{F}$ which fall within a radius of $2R$ of $p$) is maximized.
The vertices corresponding to these points (including $p$) are made into a cluster, and $p$ is remembered as the location of the cluster in the embedding. Next, $p$ and its near-by points are removed from $\mathcal{F}$.
This iterative process continues either for $k$ iterations, or until there are no points of $\mathcal{F}$ left in the embedding.
Afterward, any remaining points of $\mathcal{F}$ are thought of as ``outliers'', and each has its corresponding vertex assigned to a nearest cluster.

A more formal description of the algorithm appears in Figure \ref{fig:algorithm}. 
In Section~\ref{fast-algo}, we show how the algorithm can be implemented in time $\tilde{O}(mk+\eps^{-1} k^3 n)$ for any error parameter $\eps>0$, where $m$ denotes the number of edges of the graph and $\tilde{O}(\cdot)$ hides $\polylog n$ factors.

%Inductively define a sequence of subsets ${\cal P} = \{P_1,\ldots,P_k\}$ of $V(G)$ that uses an auxiliary sequence $V(G) = V_0 \supseteq V_1 \supseteq \ldots \supseteq V_k$. For any $i\in\{1,\ldots,k\}$ and a chosen $R>0$, we proceed as follows. For any $u\in V_{i-1}$, let
%
% \[
%  N_i(u) := \f^{-1}(\ball(\f(u), 2R)) \cap V_{i-1} = \{w\in V_{i-1} :
%    \|\f(u)-\f(w)\|_2 \leq 2R\}
% \]
%
% and let $u_i\in V_{i-1}$ be a vertex maximizing $|N_i(u)|$. We set $P_i = N_i(u_i)$, and $V_i = V_{i-1} \setminus P_i$. This chooses the $k$-centers $U=\{u_1,\ldots,u_k\}$. Next, we assign each left over vertex $v\in V_k$ to to its nearest neighbor in $U$, thus creating the output clusters $\{C_1,\ldots, C_k\}$. Formally, there is a map $\nu: V_k\rightarrow U$ where $\nu(v)=u_i$ if $i$ is the smallest index satisfying $\|\f(v)-\f(u_i)\| \leq \|\f(v)-\f(u_j)\|$ for all $j\not=i$. The final partition of $V(G)$ is $\{C_1,\ldots,C_k\}=\{P_1\cup \nu^{-1}(u_1),\ldots,P_k\cup\nu^{-1}(u_k)\}$.

\begin{figure}
\begin{center}
\begin{tabularx}{\textwidth}{l}
\toprule
\textbf{Algorithm: Greedy Spectral $k$-Clustering}\\
\textbf{Input: $n$-vertex graph $G$}\\
\textbf{Output: Partition ${\cal C}=\{C_1,\ldots,C_k\}$ of $V(G)$}\\
\midrule
Let $\vv_1,\ldots,\vv_k$ be the $k$ first eigenvectors of $\LL_G$. \\
Let $\f:V(G)\to \mathbb{R}^k$, where $\f(u)=\deg(u)^{-1/2}\left(\vv_1(u), \ldots, \vv_k(u)\right)$. \\
$R=\frac{1}{36k\dmax\sqrt{n}}$ \\
$V_0 = V(G)$ \\
for $i=1,\ldots,k$ \\
~~~~~~~~$u_i = \argmax_{u\in V_{i-1}}|\ball(\f(u),2R)\cap \f(V_{i-1})| $\\
~~~~~~~~~~~~$= \argmax_{u\in V_{i-1}}|\{w\in V_{i-1} : \|\f(u)-\f(w)\|_2 \leq 2R\}|$ \\
~~~~~~~~$P_i = \f^{-1}(\ball(\f(u_i),2R))\cap V_{i-1}$ \\
~~~~~~~~$V_i = V_{i-1} \setminus P_i$ \\
Let $\nu: V_k\rightarrow \{u_1,\cdots,u_k\}$, $\nu(v)=u_i$ if $i$ is the smallest index \\
~~~~~~~~~~~~satisfying $\|\f(v)-\f(u_i)\| \leq \|\f(v)-\f(u_j)\|$ for all $j\not=i$. \\
Return $\{C_1,\ldots,C_k\}=\{P_1\cup \nu^{-1}(u_1),\ldots,P_k\cup \nu^{-1}(u_k)\}$ \\
\bottomrule
\end{tabularx}
\end{center}
\caption{The greedy spectral $k$-clustering algorithm takes an $n$-vertex graph
$G$ with maximum degree at most $\dmax$ as input, and
outputs a partition ${\cal C}=\{C_1,\ldots,C_k\}$ of $V(G)$.
\label{fig:algorithm}}
\end{figure}

To measure the performance of our algorithm, we introduce the following notion
of symmetric difference between two collections, each of $k$ sets, that
generalizes the symmetric difference between two sets.

\paragraph*{A distance on $k$-partitions} For two sets $Y,Z$, their
symmetric difference is given by $ Y\mytriangle Z = (Y\setminus Z) \cup
(Z\setminus Y). $ Let $X$ be a finite set, $k\geq 1$, and let ${\cal
A}=\{A_1,\ldots,A_k\}$, ${\cal A}'=\{A_1',\ldots,A_k'\}$ be collections of
disjoint subsets of $X$. Then, we define a distance function between ${\cal A}$,
${\cal A}'$, by
\[
|{\cal A} \mytriangle {\cal A}'|
= \min_{\sigma} \sum_{i=1}^k \left|A_i \triangle A_{\sigma(i)}'\right|,
\]
where $\sigma$ ranges over all permutations $\sigma$ on $\{1,\ldots,k\}$.

\subsection{Main Theorem}

\begin{theorem}[Spectral partitioning via greedy $k$-clustering]\label{thm:main}
Let $G$ be an $n$-vertex graph with maximum degree at most $\dmax$. Let $k\geq 1$, and let ${\cal
A}=\{A_1,\ldots,A_k\}$ be any $(\alpha_{\myin}, \alpha_{\myout})$-strong
$k$-partition of $V(G)$ with $\alpha_{\myin}> 10 (k\dmax)^{3/2}\cdot
\sqrt{\lambda_k}$. Then, on input $G$, the algorithm in Figure
\ref{fig:algorithm} outputs a partition ${\cal C}$ such that
\[
  |{\cal A} \mytriangle {\cal C}| = O\left(\frac{\lambda_k\cdot \dmax^3k^5\cdot n}{\alpha_{\myin}^2}\right).
\]
\end{theorem}

Due to the dependency on $\dmax$, the above result is mainly interesting for bounded degree graphs. We remark that even for some special classes of bounded degree graphs, analyzing the performance of spectral clustering algorithms is already interesting and challenging. For example, in their seminal work Spielman and Teng gave the first rigorous analysis of the performance of spectral clustering methods which use the second eigenvector of the matrix $D-A$ on bounded degree planar graphs and finite element meshes~\cite{ST2007}. This result was further generalized to graphs with bounded degree and bounded genus by Kelner~\cite{DBLP:journals/siamcomp/Kelner06} and excluded-minor graphs by Biswal, Lee and Rao~\cite{DBLP:journals/jacm/BiswalLR10}. Our result holds for arbitrary bounded degree graphs that admit a good quality $k$-partition and demonstrates the effectiveness of the spectral clustering algorithms that use only the first $k$ eigenvectors (cf. Alpert and Yao~\cite{DBLP:conf/dac/AlpertY95}).

We further remark that while $\alpha_\myout$ does not appear explicitly in the error
term in Theorem~\ref{thm:main}, it does implicitly bound the error through a
higher order Cheeger inequality~\cite{DBLP:conf/stoc/LeeGT12}. In particular,
$\lambda_k\leq 2\alpha_{\myout}$ and thus when
${\alpha_{\myout}}/{\alpha^2_{\myin}}$ is small there is strong agreement
between ${\cal A}$ and ${\cal C}$. In addition, the dependency on the upper bound $\dmax$ of maximum vertex degree seems unavoidable since we are measuring
the \emph{size} of the symmetric difference ${\cal A} \mytriangle {\cal C}$
rather than its \emph{volume} or \emph{total degree}. (Note that the definition of conductance is volume-based, while the error is measured with respect to the size of clusters. Such an inconsistency seems to cause the dependency on $\dmax$.) In contrast, the latter
measurement was used in~\cite{DBLP:conf/colt/PengSZ15}, which allows the
authors to derive an error term that is independent of the maximum degree. On the other hand, such a dependency also does not
appear in~\cite{DBLP:conf/innovations/AwasthiCKS16} and~\cite{DBLP:conf/soda/Sinop16}, as the authors are
studying the size-based definition of conductance (i.e.,
$\phi_\myout(S;G):=\frac{|E(S,V(G)\setminus S)|}{|S|}$) instead of the
volume-based definition as in our paper. We also note that the algorithms in some follow-up work (e.g.,~\cite{DBLP:conf/innovations/AwasthiCKS16,DBLP:conf/colt/PengSZ15}) can output some partition with an approximate guarantee
for individual clusters, while our algorithm can only have approximate guarantee over the all $k$ clusters. 

\paragraph*{Application of main theorem} Oveis Gharan and Trevisan
\cite{DBLP:conf/soda/GharanT14} (see also \cite{Tanaka:2011vk}) showed that, if
the \emph{gap} between $\lambda_k$ and $\lambda_{k+1}$ is large enough, then
there exists a partition into $k$ clusters, each having small external
conductance and large internal conductance.

\begin{theorem}[\cite{DBLP:conf/soda/GharanT14}]\label{thm:OT_clustering} There
exist universal constants $c>0$, $\alpha>0$, and $\beta>0$, such that for any
graph $G$ with $\lambda_{k+1} > c k^2 \sqrt{\lambda_k}$, there is a
$(\alpha\cdot\lambda_{k+1}/k, \allowbreak \beta\cdot k^3
\sqrt{\lambda_k})$-strong $k$-partition of $G$.\label{OT-thm}
\end{theorem}

The same paper \cite{DBLP:conf/soda/GharanT14} also shows how to compute a
partition with slightly worse quantitative guarantees, using an iterative
combinatorial algorithm with polynomial running time. More specifically, they have shown the following theorem.
\begin{theorem}[\cite{DBLP:conf/soda/GharanT14}]\label{thm:OT_clustering_algorithm}
There is a polynomial time algorithm that takes as input a
graph $G$ with $\lambda_{k+1}>0$ for any $k\geq 1$, outputs an
$\ell$-partition that is $(\Omega(\lambda_{k+1}^2/k^4),
O(k^6\sqrt{\lambda_{k}}))$-strong, for some $1\leq\ell <k+1$.
\end{theorem}

Let $\tau$ be a number satisfying $\tau\geq \tau_0:=\max\{c, 10\cdot\alpha^{-1}\cdot
\sqrt{k}\dmax^{3/2}\}$, where $\alpha, c$ are the constants given
in Theorem \ref{thm:OT_clustering}. Let $G$ be a graph with $\lambda_{k+1}>\tau k^2\sqrt{\lambda_k}$. By applying Theorem~\ref{thm:main} on $G$ with parameters $\alpha_{\myin}=\alpha\cdot \lambda_{k+1}/k$, which satisfies that $\alpha_{\myin}> 10(k\dmax)^{3/2}\cdot \sqrt{\lambda_k}$, we obtain the following corollary. %of our main theorem.
\begin{corollary}\label{cor:main}
Let $k\geq 1$. Let $G$ be an $n$-vertex graph with maximum degree at most $\dmax$, and $ \lambda_{k+1} > \tau k^2
\sqrt{\lambda_{k}} $, where $\tau\geq \tau_0$ and $\tau_0$ is defined as above. Let ${\cal A}$ be the $(\alpha\cdot \lambda_{k+1}/k,\beta\cdot k^3\sqrt{\lambda_k})$-strong partition of $G$
guaranteed by Theorem \ref{thm:OT_clustering}. Then, on input $G$, the algorithm
in Figure \ref{fig:algorithm} outputs a partition ${\cal C}$ such that
\[
  |{\cal A} \mytriangle {\cal C}| = O\left(\frac{\dmax^3 k^3 n}{\tau^2}\right).
\]
\end{corollary}
In comparison with the algorithm from Theorem~\ref{thm:OT_clustering_algorithm} that finds a partition that is a $(\Omega(\lambda_{k+1}^2/k^4),
O(k^6\sqrt{\lambda_{k}}))$-strong partition, our algorithm finds a partition that is \emph{close to} some $(\Omega(\lambda_k/k),O(k^3\sqrt{\lambda_k}))$-strong partition. It is not clear how to find in polynomial time a partition (without error) that is $(\Omega(\lambda_k/k),O(k^3\sqrt{\lambda_k}))$-strong.  

\section{Spectral Concentration}\label{sec:concentration} In this section, we
prove that for any graph with some strong $k$-partition and any eigenvector $\vv_i$ ($1\leq i\leq k$), the rescaled vector
$\xx_i:=\vv_i\mat{D}^{-1/2}$ is close (with respect to the $\ell_2$ norm) to
some vector $\widetilde{\xx}_i$ that is constant on each cluster. We slightly
abuse the notation by also using $\f$ to denote the $n\times k$ matrix that
corresponds to our spectral embedding (i.e., with row vectors $\f(u)$, for all
$u\in V$). It is useful to note that $\f = [\xx_1^T,\cdots,\xx_k^T]$.

\begin{lemma}\label{lem:uniformity}
Let $G$ be a graph with maximum degree at most $\dmax$.
Let $\vv_1,\ldots,\vv_k \in \mathbb{R}^k$ denote the first $k$ eigenvectors of $\LL_G$.
For $\alpha_{\myin}>0$, let ${\cal
A}=\{A_1,\ldots,A_k\}$ be any $(\alpha_{\myin}, \alpha_{\myout})$-strong
$k$-partition of $V(G)$.
For any $i \in \{1,\ldots,k\}$, if
$\xx_i=\vv_i\mat{D}^{-1/2}$, then there exists $\widetilde{\xx}_i \in
\mathbb{R}^n$, such that,
\begin{description}
  \item{(i)}  $\|\xx_i-\widetilde{\xx}_i\|_2^2 \leq \frac{2 k \lambda_k \cdot \dmax}{\alpha_{\myin}^2}$, and
  \item{(ii)} $\widetilde{\xx}_i$ is constant on the clusters of ${\cal A}$, i.e.~for any $A \in {\cal A}$, $u,v\in A$, we have $\widetilde{\xx}_i(u)=\widetilde{\xx}_i(v)$.
\end{description}
\end{lemma}

Before laying out the proof, we provide some explanation of the statement of the
theorem. First, note that the $\ell_2^2$-distance between $\xx_i$ and its
uniform approximation $\widetilde{\xx}_i$ depends linearly on the ratio
$\lambda_k/\alpha_{\myin}^2$, which, as noted above, is bounded from above by
$2\alpha_{\myout}/\alpha_{\myin}^2$. Second, the partition-wise uniform vector
$\widetilde{\xx}_i$ which minimizes the left hand side of (i) is constructed by
taking the mean values of $\xx_i$ on each partition. This, together with the
bound in (i), means that $\xx_i$ assumes values in each partition close to their
mean. In summary, if there is a sufficiently large gap between the external
conductance $\alpha_{\myout}$ and internal conductance $\alpha_{\myin}$ of the
clusters, the values taken by each vector $\xx_i$ have $k$ prominent modes over
$k$ partitions.

We need the following result that is a slight restatement of a lemma in
\cite{DBLP:conf/stoc/CzumajPS15} to prove Lemma \ref{lem:uniformity}. (For completeness, a
proof of Lemma~\ref{lem:clusterable-eigenvector} is included in~\ref{appendix:eigenvector}.)

\begin{lemma}
\label{lem:clusterable-eigenvector}
Let $G = (V,E)$ be any undirected graph and let $C \subseteq V$ be any subset with
$\phi(G[C])\geq\phi_{\myin}>0$. Then for every $i$, $1 \le i \le k$,
$\xx_i=\vv_i\mat{D}^{-1/2}$, the following holds:
\[
\sum_{u, v \in C} (\xx_i(u) - \xx_i(v))^2
  \leq \frac{4\lambda_k \cdot\vol(C)}{\phi_{\myin}^2}.
\]
\end{lemma}

We remark that in the above Lemma, there is a linear dependency on $\vol(C)$, which directly causes our result and the following analysis to depend on $\dmax$. Now we are ready to prove Lemma~\ref{lem:uniformity}.
\begin{proof}[Proof of Lemma~\ref{lem:uniformity}]
Let $1\leq i\leq k$, and $1\leq j\leq k$. By precondition of the lemma, ${\cal A}$ is an $(\alpha_{\myin},\alpha_{\myout})$-strong
$k$-partition. Now we apply $C=A_j$ and $\phi_{\myin}=\alpha_{\myin}$ in
Lemma~\ref{lem:clusterable-eigenvector} to get
\begin{displaymath}
    \sum_{u, v \in A_j} (\xx_i(u) - \xx_i(v))^2
	\le
       \frac{4\lambda_k \cdot\vol(A_j)}{\alpha_{\myin}^2}
       \leq
       \frac{4\lambda_k \cdot \dmax\cdot|A_j|}{\alpha_{\myin}^2},
\end{displaymath}
where the second inequality follows from our assumption that the maximum degree is $\dmax$.
Let $\widetilde{\xx}_i(u)=\frac{\sum_{v\in A_j}\xx(v)}{|A_j|}$ if $u\in A_j$. Note that $\widetilde{\xx}_i$ is constant on each cluster. On the other hand, by the definition of $\widetilde{\xx}_i$, we have
\begin{displaymath}
\frac{1}{|A_j|}\sum_{u, v \in A_j} (\xx_i(u) - \xx_i(v))^2
=2\sum_{u\in A_j}(\xx_i(u) - \widetilde{\xx}_i(u))^2.
\end{displaymath}

Therefore,
\begin{eqnarray*}
\norm{\xx_i - \widetilde{\xx}_i}_2^2 &=& \sum_{j=1}^k\sum_{u\in A_j}(\xx_i(u) - \widetilde{\xx}_i(u))^2
     =      \frac12\sum_{j=1}^k\frac{1}{|A_j|}\sum_{u, v \in A_j} (\xx_i(u) - \xx_i(v))^2 \\
     &\leq& \frac{2 k \lambda_k \cdot \dmax}{\alpha_{\myin}^2}.
\end{eqnarray*}
%where the penultimate inequality follows from our assumption that $\lambda_{k+1} \geq \tau k^2\sqrt{\lambda_k}$.
This completes the proof of the lemma.
\end{proof}

\section{From Spectral Concentration to Spectral Clustering}\label{sec:clustering}

In this section we prove Theorem \ref{thm:main}. We begin by showing that if
there exists strong $k$-partition with high quality in the graph $G$, then in
the spectral embedding defined by $\f(u) =\deg(u)^{-1/2}
(\vv_1(u),\ldots,\vv_k(u))$ for any $u\in V$, one can find $k$ well-separated
center points in $\mathbb{R}^k$ such that the balls (of some appropriately
chosen radius) centered at these center points are disjoint and the collection
of these balls is close to any strong $k$-partition.

\begin{lemma}\label{lem:sep}
Let $G$ be an $n$-vertex graph of maximum degree at most $\dmax$.
Let ${\cal A}=\{A_1,\ldots,A_k\}$ be any $(\alpha_{\myin},
\alpha_{\myout})$-strong $k$-partition of $V(G)$ with $\alpha_{\myin}> 10
(k\dmax)^{3/2}\cdot \sqrt{\lambda_k}$. Let $\f:V(G)\to \mathbb{R}^k$ be
the spectral embedding of $G$ given by (\ref{eq:embedding}). Let
$R=\frac{1}{36k\dmax\sqrt{n}}$. Then there exists $k$ points
$\pp_1,\ldots,\pp_k\in \mathbb{R}^k$ and a family, ${\cal A}'$, of $k$ subsets of $V(G)$, given by $A_i' = \{u\in V: \norm{\f(u)-\pp_i}_2\leq R\}$ for $1 \leq i \leq k$, such that the following conditions are satisfied:
\begin{description}
\item{(i)} For $1 \leq i < j \leq k$, $\|\pp_i - \pp_j\|_2 > 6 R$.
\item{(ii)} The elements of ${\cal A}'$ are pairwise disjoint.
\item{(iii)} $|{\cal A} \mytriangle {\cal A}' | = O\left(\frac{\lambda_k\cdot \dmax^3k^4\cdot n}{\alpha_{\myin}^2}\right)$.
\end{description}
\end{lemma}

To prove Lemma~\ref{lem:sep}, we first give some definitions and introduce some useful tools. 
For any symmetric matrix $\mat{X}$, let $\eta_i(\mat{X})$ denote the $i$th
largest eigenvalue of $\mat{X}$. For any (not necessarily square) matrix
$\mat{Y}$, let $\mat{Y}_{\row(i)}$ denote the $i$th row vector of $\mat{Y}$.
We will make use of the following pair of facts which are proved in~\ref{sec:facts}:

\begin{fact}\label{fact:eigenvalue-A-B}
For any two $p\times p$ symmetric matrices $\mat{X},\mat{Y}$, if $\max_{i\leq
p}\norm{\mat{X}_{\row(i)}-\mat{Y}_{\row(i)}}_2\leq \delta$,  then for any $i\leq
p$, $|\eta_i(\mat{X})-\eta_i(\mat{Y})|\leq \sqrt{p}\cdot\delta$.
\end{fact}

\begin{fact}\label{fact:rowdifference}
For any two $p\times q$ matrices $\mat{X},\mat{Y}$, if $\max_{i\leq
p}\norm{\mat{X}_{\row(i)}}_2\leq \gamma$, and $\max_{i\leq
p}\norm{\mat{X}_{\row(i)}-\mat{Y}_{\row(i)}}_2 \leq \delta$, then $\max_{i\leq
p}\norm{(\mat{X}\cdot \mat{X}^T)_{\row(i)}-(\mat{Y}\cdot \mat{Y}^T)_{\row(i)}}_2
\leq \sqrt{p}(\delta^2 + 2\gamma \delta)$.
\end{fact}

Now we prove Lemma~\ref{lem:sep}.
\begin{proof}[Proof of Lemma~\ref{lem:sep}]
Recall that for any $i$, $1 \le i \leq k$, $\xx_i = \vv_i\mat{D}^{-1/2}$ and
$\widetilde{\xx}_i$ denotes the vector that $\widetilde{\xx}_i(u) =
\frac{1}{|A_j|} \sum_{v \in A_j} \xx_i(v)$ if $u \in A_j$. Now for each $1\leq
j\leq k$, define
\[
  \pp_j:=(\widetilde{\xx}_1(u),\cdots,\widetilde{\xx}_k(u)), \quad \textrm{for any $u\in A_j$.}
\]
Further recall that $\f\in \mathbb{R}^{n\times k}$ is the matrix corresponding
to our spectral embedding and that $\f = [\xx_1^T,\cdots,\xx_k^T]$. Let
$\mat{P}:=[\widetilde{\xx}_1^T,\cdots,\widetilde{\xx}_k^T]$.
%Let us consider its transposed matrix $\mat{X}:=\f^T$. Note that $\mat{X}\in \mathbb{R}^{k\times n}$ and
%$\mat{X}_{\row(i)}=\xx_i$ for each $i\leq k$. Let $\mat{\XX}$ be the $k\times n$
%matrix with $\mat{\XX}_{\row(i)}=\widetilde{\xx}_i$ for each $i\leq k$.
Note that for any $u\in A_j$, the row vector corresponding to $u$ of $\mat{P}$
is $\pp_j$. Let $\zeta=\frac{2 k \lambda_k \cdot
\dmax}{\alpha_{\myin}^2}$.
We have the following claim about the eigenvalues of matrix $\mat{P}^T\cdot \mat{P}$.

\begin{claim}\label{lem:eigen_lower}
All eigenvalues of $\mat{P}^T\cdot \mat{P}$ are at least
$\frac{1}{\dmax}-k(\zeta+2\sqrt{\zeta})$.
\end{claim}

\begin{proof}
Let $\mat{T}$ be the $n\times k$ matrix with $i$th column $\vv_i^T$, for each
$i\leq k$. Thus, $\mat{F}=\mat{D}^{-1/2}\mat{T}$. Now we note that all the
eigenvalues of $\f^T\cdot \f$ are at least $1/\dmax$. This is true since for any
$\mat{y}\in\mathbb{R}^k$,
\begin{eqnarray*}
\mat{y}(\f^T\cdot \f)\mat{y}^T
    = \norm{\mat{y}\f^T}_2^2
    = \norm{\mat{y} \mat{T}^T \mat{D}^{-1/2}}_2^2
 \geq \norm{\mat{y} \mat{T}^T}_2^2/\dmax
    = \mat{y}\cdot \mat{y}^T / \dmax,
\end{eqnarray*}
where the second to last inequality follows from the fact that
$\norm{\mat{z}\mat{M}}_2^2\geq (\min_{i}\mat{M}_{i,i})^2\norm{\mat{z}}_2^2$ for
any vector $\mat{z}\in \mathbb{R}^n$ and diagonal matrix $\mat{M}\in
\mathbb{R}^{n\times n}$, which is true since
$\norm{\mat{z}\mat{M}}_2^2
    = \sum_{i}(\mat{z}_i\mat{M}_{i,i})^2
 \geq (\min_{i}\mat{M}_{i,i})^2\sum_{i}(\mat{z}_i)^2
    = (\min_{i}\mat{M}_{i,i})^2\norm{\mat{z}}_2^2$%
; and the last equation follows from the observation that $\mat{T}^T\cdot \mat{T}=\mat{I}_{k\times k}$.

Now note that for each $i\leq k$,
since $\xx_i = \vv_i\mat{D}^{-1/2}$, and $\norm{\vv_i}_2=1$, it holds that
$\norm{\xx_i}_2\leq 1$. On the other hand, by
Lemma~\ref{lem:uniformity}, it holds that
%$\norm{X_{\row(i)}}_2=\norm{\xx_i}_2\leq 1$ and
\begin{eqnarray*}
\norm{\mat{F}^T_{\row(i)}-\mat{P}^T_{\row(i)}}_2
    = \norm{\xx_i-\widetilde{\xx}_i}_2 \leq \sqrt{\zeta}.
\end{eqnarray*}

By Fact~\ref{fact:rowdifference}, we have that
$\max_{i\leq p}\norm{(\mat{F}^T\cdot \mat{F})_{\row(i)}-(\mat{P}^T\cdot \mat{P})_{\row(i)}}_2 \leq \sqrt{k}(\zeta+2\sqrt{\zeta})$.
Then by Fact~\ref{fact:eigenvalue-A-B}, for each $i\leq k$,
$|\eta_i(\mat{F}^T\cdot \mat{F})-\eta_i(\mat{P}^T\cdot\mat{P})|\leq
k(\zeta+2\sqrt{\zeta})$. Since all the eigenvalues of
$\mat{F}^T\cdot \mat{F}$ are at least $\frac{1}{\dmax}$, it follows that that
all eigenvalues of $\mat{P}^T\cdot \mat{P}$ are at least
$\frac{1}{\dmax}-k(\zeta+2\sqrt{\zeta})$.
\end{proof}

On the other hand, we prove in the following claim that if there exists two vectors
$\pp_{i_0},\pp_{j_0}$ that are close, then $\mat{P}^T\cdot \mat{P}$ has at least
one small eigenvalue.

\begin{claim}\label{lem:eigen_upper}
If there exists $i_0,j_0\leq k$ such that $\norm{\pp_{i_0}-\pp_{j_0}}_2\leq 6R$,
then $\mat{P}^T\cdot \mat{P}$ has an eigenvalue at most
$k(36R^2n + 12R\sqrt{n})$.
\end{claim}
\begin{proof}%[Proof of Claim~\ref{lem:contradiction_eigenvalue}]
Let $\mat{Q}$ denote the $n\times k$ matrix obtained from $\mat{P}$ by
replacing each row vector that equals $\pp_{i_0}$ by vector $\pp_{j_0}$. Note that
$\mat{Q}^T\cdot\mat{Q}$ is singular, and thus has eigenvalue $0$.

Now note that
$\max_{i\leq k}\norm{\mat{P}^T_{\row(i)}-\mat{Q}^T_{\row(i)}}_2\leq 6R\sqrt{n}$ since the
absolute value of each entry in $\mat{P}-\mat{Q}$ is at most $6R$, and also note that
for any $i\leq k$,
\begin{eqnarray*}
\norm{\mat{P}^T_{\row(i)}}_2
    &=&	 \sqrt{\sum_{j=1}^k |A_j| \left(\frac{\sum_{u\in A_j} \xx_i(u)}  {|A_j|}\right)^2}
    \leq \sqrt{\sum_{j=1}^k |A_j| \frac{\sum_{u\in A_j} \xx_i^2(u)}{|A_j|}} \\
    &=&	 \norm{{\xx}_i}_2
    \leq 1.
\end{eqnarray*}

Then by Fact~\ref{fact:rowdifference},
$\max_{i\leq k}\norm{(\mat{P}^T\cdot \mat{P})_{\row(i)}-(\mat{Q}^T\cdot \mat{Q})_{\row(i)}}_2
    \leq \sqrt{k}\left(36R^2n + 2\cdot6R\sqrt{n}\right)$.
Now by Fact~\ref{fact:eigenvalue-A-B} and the fact that $\mat{Q}^T\cdot \mat{Q}$ has
eigenvalue $0$, we know that at least one eigenvalue of $\mat{P}^T\cdot \mat{P}$ is at
most $k\left(36R^2n + 12R\sqrt{n}\right)$.
\end{proof}

Now we prove Item~(i) of the lemma. By the
assumption $\alpha_{\myin} > 10 (k\dmax)^{3/2} \sqrt{\lambda_k}$, it holds
that $\zeta=\frac{2 k \lambda_k \cdot
\dmax}{\alpha_{\myin}^2}<\frac{1}{50k^2\dmax^2}$, which implies
that
$\frac{1}{\dmax}-k(\zeta+2\sqrt{\zeta})>\frac{1}{2\dmax}$.
On the other hand, since $R=\frac{1}{36k\dmax\sqrt{n}}$, we have that
$k(36R^2n + 12R\sqrt{n}) \leq \frac{1}{2\dmax}$. Therefore,
by Claim~\ref{lem:eigen_lower} and Claim~\ref{lem:eigen_upper}, we have reached
a contradiction regarding the minimum eigenvalue of $\mat{P}^T\cdot\mat{P}$.
This implies Item (i) of the lemma, that is, for all $1\leq i<j\leq k$,
$\norm{\pp_i-\pp_j}_2 > 6R$. Now for each $1\leq j\leq k$, define $A_j':=\{u\in
V: \norm{\f(u)-\pp_j}_2\leq R\}$ as required by Item (ii). Let ${\cal
A}'=\{A_1',\ldots,A_k'\}$. By Item (i) $A_1',\cdots,A_k'$ are disjoint, proving
Item (ii).

Finally, we prove Item (iii) of the lemma. By
Lemma~\ref{lem:uniformity},
\[
\sum_{i=1}^k\norm{\xx_i-\widetilde{\xx}_i}_2^2 \leq k\cdot \zeta.
\]
On the other hand, if we let $A_\bad=\{u:u\in A_j, \norm{\f(u)-\pp_j}>R, 1\leq
j\leq k\}_2$, then
\begin{eqnarray*}
  \sum_{i=1}^k\norm{\xx_i-\widetilde{\xx}_i}_2^2
    =    \sum_{j=1}^k\sum_{u\in A_j}\sum_{i=1}^k(\xx_i(u)-\widetilde{\xx}_i(u))^2
   &=&   \sum_{j=1}^k\sum_{u\in A_j}\norm{\f(u)-\pp_j}_2^2\\
&\geq&   \sum_{u\in A_\bad}R^2=|A_\bad|\cdot R^2.
\end{eqnarray*}
Therefore,
$|A_\bad|\leq\frac{k\zeta}{R^2}
\leq \frac{2592\cdot\lambda_k\cdot \dmax^3k^4\cdot n}{\alpha_{\myin}^2}$. 

Now we observe that 
$$|{\cal A} \mytriangle {\cal A}'| \leq \sum_{j=1}^k|A_j\mytriangle A_j'|=\sum_{j=1}^k(|A_j \setminus A_j'| + |A_j'\setminus A_j|)\leq |A_\bad| + \sum_{j=1}^k|A_j'\setminus A_j|.$$

Note that for any $u\in A_j'\setminus A_j$, it holds that $u\in A_i$ for some $i\neq j$. Since $u\in A_j'$, it holds that $\norm{\f(u)-\pp_j}_2\leq R$, which implies that $\norm{\f(u)-\pp_i}_2\geq \norm{\pp_j-\pp_i}_2-\norm{\f(u)-\pp_j}_2>6R-R=5R$. This further implies that $u\in A_\bad$. Since all $A_1',\cdots,A_k'$ are disjoint, we have that $\sum_{j=1}^k|A_j'\setminus A_j|\leq |A_\bad|$.

Thus, it holds that 
$$|{\cal A} \mytriangle {\cal A}'|\leq 2|A_{\bad}|,$$
which proves the Item (iii) of the lemma.
\end{proof}

We are now ready to prove our main theorem.

\begin{proof}[Proof of Theorem \ref{thm:main}]
Let ${\cal A}=\{A_1,\ldots,A_k\}$, ${\cal A}'=\{A_1',\ldots,A_k'\}$, $R$, and $\pp_1,\ldots,\pp_k$ be as in
Lemma \ref{lem:sep}. Let
    $\eps =|{\cal A} \mytriangle {\cal A}'|/n
	= O(\frac{\lambda_k\cdot \dmax^3k^4}{\alpha_{\myin}^2})$. %Note that $|B|\leq |{\cal A}\mytriangle {\cal A'}|= \eps  n$.
Let ${\cal C}=\{C_1,\ldots,C_k\}$ be the ordered collection of
pairwise disjoint subsets of $V(G)$ output by
the greedy spectral $k$-clustering algorithm
in Figure~\ref{fig:algorithm}. Let ${\cal P} = \{P_1,\ldots,P_k\}$ where $P_i$ is the subset, called \emph{group}, found by the algorithm at the $i$th iteration for $1\le i\leq k$.
The set of vertices not covered by any of the clusters in ${\cal A}'$
plays a special role in our argument which we denote as
$B=V(G)\setminus \left(\bigcup_{i=1}^k A_i'\right)$. Clearly,
$|B|\leq |{\cal A}\mytriangle {\cal A'}|\leq \eps  n$.

We say that a cluster $A_i'$ is {\em touched}
if the algorithm, while computing the centers,
considers a group $P_j\in {\cal P}$ with
$P_j\cap A_i'\neq \emptyset$.
%Since the clusters in ${\cal C}$ cover all vertices in $V(G)$, every cluster in ${\cal A}'$ is touched by some cluster in ${\cal C}$.
For a cluster $A_i'$, let $P_{\rho(i)}$ be the group
in ${\cal P}$ that touches $A_i'$ for the {\em first time} in the
algorithm if it is touched at all.
Let $I\subseteq \{1,\ldots,k\}$ be the support
of $\rho$, that is, $\rho(i)$ exists if and only if $i\in I$.
Let $i^*=|I|$.
By permuting the indices of the clusters in ${\cal A}'$, we may assume w.l.o.g.~that $I = \{1,\ldots,i^*\}$.

First we observe that $\rho$ is a bijection on $I$.
%In particular, if $i^*=k$, then $\rho$ is a bijection between $\{1,\ldots,k\}$ and $\{1,\ldots,k\}$.
%${\cal C}$ and ${\cal A}'$.
%If $i^*<k$, we claim that the clusters $\{A_i'\,|\,i=i^*+1,\ldots,k\}$
%are all mapped to $C_k$, that is, $\rho(i^*+1)=\ldots=\rho(k)=k$.
This is because the group $P_{\rho(i)}$, $i\in I$, can intersect at most one
cluster in ${\cal A}'$. The reason is that every cluster in ${\cal A}'$ is
contained inside some ball of radius $R$, the distance between any two centers
of such balls is more than $6R$, and each $P_{\rho(i)}$ is contained inside some
ball of radius $2R$.

In case $i\in I$, we have
$|P_{\rho(i)}\setminus A_i'| \leq |B\cap P_{\rho(i)}|$.
%This is certainly true if $\rho(i)=k$ because then $C_k$ contains
%$A_i'$ completely.
%When $\rho(i)\not=k$,
This is because $P_{\rho(i)}$ cannot intersect any other cluster in ${\cal A}'$
but $B$. On the other hand, it holds that $|A_i'\setminus P_{\rho(i)}|\leq \eps
n$, since otherwise the algorithm could have made a better choice by taking the
entire $A_i'$ while computing $P_{\rho(i)}$. Such a choice can be made by taking
$P_{\rho(i)}$ to be all the yet unclustered points that are inside a ball of
radius $2R$ centered at any point in $A_i'$; since $A_i'$ is in a ball of radius
$R$, it follows by the triangle inequality that $A_i'$ will be contained inside
$P_{\rho(i)}$. Therefore, for every $i \leq i^*$, $|A_i'\mytriangle P_{\rho(i)}| =
|P_{\rho(i)}\setminus A_i'| + |A_i'\setminus P_{\rho(i)}|\leq |B\cap
P_{\rho(i)}|+\eps n$.

In the other case when $i\not \in I$, we claim that the cluster $A_i'$ can have
at most $2\eps n$ vertices. Suppose not. Since $\rho$ is a bijection on $I$ and
$I$ is a proper subset of $\{1,\cdots,k\}$, there is a group $P_j$ with
$j\not\in \rho(I)$, which does not intersect any cluster in ${\cal A}'$ for the
first time. Then, it has the only option of intersecting a cluster in ${\cal
A}'$ beyond the first time and/or intersect $B$. Since $|A_i'\setminus
P_{\rho(i)}|\leq \eps n$ for all $i\in I$, $P_j$ can have at most $\eps
n+|B|\leq 2\eps n$ vertices. But, the algorithm could have made a better choice
by selecting $A_i'$ while computing $P_j$ because $|A_i'|>2\eps n$ by our
assumption. We reach a contradiction.

%If $i^*<i\leq k$, $C_k$ contains $A_{i}'$ entirely and may intersect
%other clusters in ${\mathcal A}'$ for the second time and beyond.
%Then, we have $|A_{i}'\mytriangle C_k|\leq k\eps n+|B\cap C_{k}|$.
Let $T=\{1,\cdots, k\}\setminus \{\rho(1),\cdots,\rho(i^*)\}$ be the set of indices $j$ such that $P_j$ does not intersect any cluster
in ${\cal A}'$ for the first time. For any $j\in T$,
$P_j$ can only intersect set $B$ and/or set $A_i\setminus P_{\rho(i)}$ for
some $i$ such that $\rho(i)<j$. This then gives that
$|\cup_{j\in T} P_j|\leq \sum_{j\in T}|B\cap P_j|+\sum_{i\leq i^*} |A_i'\setminus P_{\rho(i)}|$.

Using the bijectivity of $\rho$ on the set $\{1,\ldots,i^*\}$,
we have
\begin{align*}
|{\cal A}'\mytriangle {\cal P}|
  & \leq \sum_{i\leq i^*} |A_i'\mytriangle P_{\rho(i)}| + \left|\cup_{j\in T} P_j\right| + \left|\cup_{i> i^*} A_i'\right|\\
  & \leq k\eps  n + \sum_{i\leq i^*} \left|B\cap P_{\rho(i)}\right| + \sum_{j\in T}\left|B\cap P_j\right|+k\eps  n + \left|\cup_{i> i^*} A_i'\right|\\
  & \leq 2k\eps n + |B| + 2k\eps n
    \leq 5k\eps n.
\end{align*}

Now since the vertices in $V_k = V(G)\setminus \cup_{i\leq k}P_i$ are distributed to
the clusters in ${\cal P}$ to create the output clusters ${\cal C}$, we have that
$|{\cal A'}\mytriangle {\cal C}| \leq |{\cal A'}\mytriangle {\cal P}|+|V_k|
$.
Observe that by the above analysis,
$$
|V_k|\leq |B|+
\sum_{i\leq i^*} |A_i'\setminus P_{\rho(i)}| + |\cup_{i> i^*} A_i'|
\leq \eps n + k\eps n + 2k\eps n\leq 4k\eps n.
$$
It follows that $|{\cal A'}\mytriangle {\cal C}|\leq 9k\eps n$.
The following concludes the proof:
%We conclude that
\[
|{\cal A} \mytriangle {\cal C}|
  \leq |{\cal A} \mytriangle {\cal A}'| + |{\cal A}' \mytriangle {\cal C}|
     < \eps n + |{\cal A}' \mytriangle {\cal C}|
  \leq 10k\eps n
     = O\left(\frac{\lambda_k\cdot \dmax^3k^5\cdot n}{\alpha_{\myin}^2}\right).\]
\end{proof}

\section{Implementation in Practice}\label{fast-algo}
In this section, we show how to efficiently implement the greedy algorithm in
Figure~\ref{fig:algorithm}. To this end, we discuss how to quickly compute the
first $k$ eigenvectors and how to speed up the step of finding centers by random
sampling.

\subsection{Eigenvectors Computation} In general, the eigenvectors cannot be
computed exactly in polynomial time as the entries may be irrational. However,
in our application it is sufficient to have a set of vectors that well
approximate the eigenvectors. To see this, we note that only the orthonormal
property of eigenvectors $\vv_1,\cdots,\vv_k$ and the fact that for $1\leq i\leq
k$, $\frac{\vv_i\LL_G\vv_i^T}{\vv_i\cdot\vv_i^T} \leq \lambda_i$ are needed for
all our previous results and analysis. Therefore, if we have a set of $k$
orthonormal vectors $\vvv_1,\cdots,\vvv_k$ with
$\frac{\vvv_i\LL_G\vvv_i^T}{\vvv_i\cdot\vvv_i^T} \leq 2\lambda_i$ for $1\leq
i\leq k$, then our previous results still hold if we replace $\lambda_k$ by
$2\lambda_k$. On the other hand, such set of $k$ orthonormal vectors can be
computed efficiently as shown in the following folklore lemma, the proof of
which follows from a repeated application of the near-linear time algorithm for
computing the second eigenvector given by Spielman and Teng~\cite{DBLP:journals/siammax/SpielmanT14}
and the variational characterization of eigenvalues (see e.g., Corollary 7.6.4
in \cite{Ove14:rounding}).

\begin{lemma}[\bf{folklore}]\label{lem:approx-eigen}
There exists a procedure, \textup{\textbf{ApproxEigen}}, that takes an
$n$-vertex, $m$-edge graph $G$, and an integer $k\leq n$, and returns $k$
orthonormal vectors $\vvv_1,\cdots,\vvv_k \in \mathbb{R}^n$ such that
$\frac{\vvv_i\LL_G\vvv_i^T}{\vvv_i\cdot\vvv_i^T} \leq 2\lambda_i$, for $1\leq
i\leq k$. The running time of the procedure is $\tilde{O}((m+n)k)$.
\end{lemma}

\subsection{A Faster Algorithm}\label{sec:fast-algorithm}
To further speed up the running time, we note that in each iteration $i$ such that $1\leq i\leq k$, the greedy algorithm in Figure~\ref{fig:algorithm} has to consider all vertices in $V_{i-1}$ to determine the best center.
This may cause the total time in these iterations to be as large as $\Omega(kn^2)$, which is slow in practice since $n$ can be much larger than $k$.
We show how to speed up this step via random sampling.
The main observation is that we can get good center candidates by only computing the number of near-by vertices in the embedding, for vertices from a randomly chosen subset $U_i$ of $V_i$, of size about $\Theta(\epsilon^{-1} k\log n)$, for any error parameter $\epsilon$.
This will reduce the computation time for finding the best centers from $\Theta(n^2 k^2)$ to $O(\epsilon^{-1} n k^3 \log n)$.
The procedure is summarized in Figure \ref{fig:algorithm2}.
The performance of the algorithm is given in the following theorem (that is similar to Theorem~\ref{thm:main}).

\begin{figure}
\begin{center}
\begin{tabularx}{\textwidth}{l}
\toprule
\textbf{Algorithm: Fast Spectral $k$-Clustering}\\
\textbf{Input: $n$-vertex graph $G$}\\
\textbf{Output: Partition ${\cal C}=\{C_1,\ldots,C_k\}$ of $V(G)$}\\
\midrule
%\begin{tabular}{ll}
%\hline
%\textbf{Algorithm: Fast Spectral $k$-Clustering}\\
%\textbf{Input: $n$-vertex graph $G$ with maximum weighted degree $\dmax$, minimum edge weight $\wmin$ and $\eps>0$.}\\
%\textbf{Output: Partition ${\cal C}=\{C_1,\ldots,C_k\}$ of $V(G)$}\\
%\hline
Let $\vvv_1,\ldots,\vvv_k$ be the returned vectors of the procedure \textbf{ApproxEigen}$(G,k)$. \\
Let $\f:V(G)\to \mathbb{R}^k$, where $\f(u)=\deg(u)^{-1/2}\left(\vvv_1(u), \ldots, \vvv_k(u)\right)$. \\
$R=\frac{1}{36k\dmax\sqrt{n}}$ \\
$V_0 = V(G)$ \\
for $i=1,\ldots,k$ \\
~~~~~~~~Sample uniformly with repetition a subset $U_{i-1}\subseteq V_{i-1}$, $|U_{i-1}| = \Theta(\eps^{-1} k\log n)$. \\
~~~~~~~~$u_i = \argmax_{u\in U_{i-1}}|\ball(\f(u),2R)\cap \f(V_{i-1})|$ \\
~~~~~~~~~~~~$ = \argmax_{u\in U_{i-1}}|\{w\in V_{i-1} : \|\f(u)-\f(w)\|_2 \leq 2R\}|$ \\
~~~~~~~~$P_i = \f^{-1}(\ball(\f(u_i),2R))\cap V_{i-1}$ \\
~~~~~~~~$V_i = V_{i-1} \setminus P_i$ \\
Let $\nu: V_k\rightarrow \{u_1,\cdots,u_k\}$, $\nu(v)=u_i$ where $i$ is the smallest index \\
~~~~~~~~~~~~satisfying $\|\f(v)-\f(u_i)\| \leq \|\f(v)-\f(u_j)\|$ for all $j\not=i$. \\
Return $\{C_1,\ldots,C_k\} = \{P_1\cup\nu^{-1}(u_1),\ldots,P_k\cup\nu^{-1}(u_k)\}$ \\
\bottomrule
\end{tabularx}
\end{center}
\caption{The fast spectral $k$-clustering algorithm takes an $n$-vertex graph
    $G$ with maximum degree at most $\dmax$, and
    an $\eps>0$ as input, and returns a partition ${\cal C}=\{C_1,\ldots,C_k\}$
    of $V(G)$. \label{fig:algorithm2}}
\end{figure}

\begin{theorem}\label{thm:algorithm2}
Let $\eps > 0$. Let $G$ be an $n$-vertex $m$-edge graph with maximum
degree at most $\dmax$. Let $k\geq 1$, and let
${\cal A}=\{A_1,\ldots,A_k\}$ be any $(\alpha_{\myin}, \alpha_{\myout})$-strong
$k$-partition of $V(G)$ with $\alpha_{\myin}>\max\{10 (k\dmax)^{3/2},\cst
\sqrt{k^5\dmax^3}\}\cdot\sqrt{\frac{2\lambda_k}{\eps}}$ for some
sufficiently large constant $\cst$. Then, on input $G$, with high probability,
the algorithm in Figure \ref{fig:algorithm2} outputs a partition ${\cal C}$ such
that
\[ |{\cal A} \mytriangle {\cal C}|\leq \varepsilon n. \]
Furthermore, the running time of the algorithm is $\tilde{O}(mk+\eps^{-1}k^3n)$.
\end{theorem}

\begin{proof}[Proof sketch]
Note that $\alpha_\myin >10 (k\dmax)^{3/2}\cdot \sqrt{2\lambda_k}$, and
that by Lemma~\ref{lem:approx-eigen} the vectors $\vvv_1,\cdots,\vvv_k$ returned
by \textbf{ApproxEigen} are orthonormal and satisfy
$\frac{\vvv_i\LL_G\vvv_i^T}{\vvv_i\cdot\vvv_i^T} \leq 2\lambda_i$, for $1\leq
i\leq k$. Thus, we can apply the argument of the proof of Lemma~\ref{lem:sep} on
vectors $\vvv_1,\cdots,\vvv_k$ to find a collection ${\cal
A}'=\{A_1',\ldots,A_k'\}$ of pairwise disjoint subsets of $V(G)$, such that
$|{\cal A} \mytriangle {\cal A}'|/n\leq O(\frac{\lambda_k\cdot
\dmax^3k^4}{\alpha_{\myin}^2})$. Since $\cst$ is sufficiently large
and $\alpha_\myin\geq \cst \sqrt{k^5\dmax^3}\cdot
\sqrt{\frac{2\lambda_k}{\eps}}$, we can guarantee that $|{\cal A} \mytriangle
{\cal A}'|/n\leq \frac{\varepsilon}{10k}$. Let $\eps'=\frac{\varepsilon}{10k}$.
Similar to the proof of Theorem~\ref{thm:main}, we define a function
$\rho:\{1,\ldots,k\} \to \{1,\ldots,k\}$ that maps clusters in ${\cal A'}$ to
clusters in ${\cal C}$, where ${\cal C}=\{C_1,\ldots,C_k\}$ are the ordered
collection of pairwise disjoint subsets of $V(G)$ output by the greedy spectral
$k$-clustering algorithm in Figure~\ref{fig:algorithm2}. Now note that for any
$i\leq k$, a vertex from a cluster $A_i'$ of size at least $2\eps' n$ will be
sampled out with probability at least $1-\frac{1}{n^2}$. This implies that with
high probability the following holds: for each $i\in \{1,\ldots,k\}$, if $|A_i'|
\geq 2\eps' n$, then $\rho(i)$ exists, since the vertices in the set $U_{i-1}$
of the $i$th iteration of the algorithm are sampled uniformly at random. The
rest of the argument for the correctness of the algorithm is identical to the
proof of Theorem \ref{thm:main}, and eventually, we can guarantee that $|{\cal
A} \mytriangle {\cal C}|\leq 10k\eps' n = \eps n$.

For the running time of the algorithm, note that by
Lemma~\ref{lem:approx-eigen}, the procedure \textbf{ApproxEigen} takes time
$\tilde{O}((m+n)k)$.
In each iteration, we need to sample $O(\eps^{-1} k\log n)$ vertices, and for each sampled vertex we need to determine the number of near-by vertices in the embedding, which takes $O(nk)$ time, as each vertex corresponds to a point in $\mathbb{R}^k$.
This means that the computation in all the $k$ iterations
takes time $O(\eps^{-1} nk^3\log n)$. Therefore, the total running time of the
algorithm is thus $\tilde{O}((m+n)k)+O(\eps^{-1} nk^3\log
n)=\tilde{O}(mk+\eps^{-1} nk^3)$.
\end{proof}

In practice it is common to work with fixed $k$. We note that the fast
randomized algorithm runs in near-linear time for $k=O(\polylog(n))$, provided
that the user specifies $\eps = \Omega(1/\polylog(n))$.

\section{Conclusion}
In this paper we have presented a very simple spectral clustering algorithm that provably approximates any good partition of an input graph, provided that one exists. Further, our preliminary experimental results given in \ref{sec:experiments} indicate that the algorithm gives meaningful output even when the  spectral gap condition is much smaller than what our theorems require.
This provides some evidence that stronger theoretical guarantees may be obtainable, possibly by weakening the gap condition. It is also natural to wonder how small of a separation is necessary for good performance of the algorithm. We currently know of no such lower bounds.
We believe that these are interesting research directions.

One interesting remark is that a qualitatively similar result can be obtained for weighted graphs by introducing the input graph's minimum edge weight, $\wmin$, as a lower bound for the (now weighted) degree in Lemma~\ref{lem:clusterable-eigenvector} and in Claim~\ref{lem:eigen_lower}.
This results in an additional factor of $\wmin^{-3}$ in the corresponding analog of Theorem~\ref{thm:main}, provided that $R$ of Figure~\ref{fig:algorithm} is also scaled by a factor of $\sqrt{\wmin}$.

Unfortunately, due to the appearance of $\wmin$ in the denominator, the result of the previous paragraph is unstable under perturbation of the input graph by a low-weight edge. This instability appears to be an artifact of the analysis. For instance we can show that given any $(\alpha_\myin, \alpha_\myout)$-strong partition $\mathcal{A}$, there exists a choice of $R$ such that the resulting clustering is stable under perturbation by a low-weight edge. Such a result essentially follows by replacing $\wmin$ with $\deltamin$, the minimum weighted vertex degree among all induced graphs $G[A]$ for $A \in \mathcal{A}$. The obstacle with turning this into an algorithm is that to obtain the corresponding error guarantee we must scale $R$ (as it appears in  Figure~\ref{fig:algorithm}) by a factor of $\sqrt{\deltamin}$, which is not known a priori. It remains an open problem to give a similar algorithm for weighted graphs which is stable under perturbation by a low-weight edge.

\section*{Acknowledgments}
The authors wish to thank James R.~Lee for bringing to their attention a result from the latest version of \cite{DBLP:conf/stoc/LeeGT12}.
This work was partially supported by the NSF grants
CCF 1318595 and CCF 1423230. The second author acknowledges the
support of ERC grant No. 307696 and a 973 program Grant No. 2014CB340302.

\bibliographystyle{plain}
\bibliography{bibfile}
\appendix

\section{Missing Proofs from Section~\ref{sec:concentration}}\label{appendix:eigenvector}

\begin{proof}[Proof of Lemma~\ref{lem:clusterable-eigenvector}]
\label{lem:clusterable-eigenvector-proof}
For any $i\leq k$,
\begin{eqnarray*}
    \frac{\vv_i \LL_G \vv_i^T}{\vv_i\vv_i^T} = \frac{\xx_i \mat{D}^{1/2}\LL_G \mat{D}^{1/2}\xx_i^T}{\xx_i \mat{D}\xx_i^T} = \frac{\sum_{(u,v)\in E}(\xx_i(u)-\xx_i(v))^2}{\sum_{u\in V}\deg(u)\xx_i^2(u)} = \lambda_i \leq \lambda_k
    \enspace.
\end{eqnarray*}
This further gives that
\begin{eqnarray}
\label{eqn:lambdaphiout}
\sum_{(u,v)\in E}(\xx_i(u)-\xx_i(v))^2 \leq \lambda_k,
\end{eqnarray}
since $\sum_{u\in V}\deg(u)\xx_i^2(u)=\sum_{u\in V}\vv_i^2(u)=1$.
Let us recall a known result (see, e.g., \cite[(1.5), p.~5 and (1.14), p.~13]{Chu97:spectral}) that for any graph $H = (V_H,E_H)$,\footnote{We remark that in \cite{Chu97:spectral}, the summation in the denominator is over all unordered pairs of vertices, while in our context, the summation is over all possible $|V_H|^2$ vertex pairs. Therefore, a multiplicative factor $2$ appears in the numerator in Equation~(\ref{ineq:from-Chung}) compared to the form in \cite[(1.5), p.~5]{Chu97:spectral}.}
\begin{equation}
\label{ineq:from-Chung}
    \lambda_2(H)
	=
    \vol_H(V_H) \cdot
    \min_{\mathbf{f}} \left\{\frac{2\cdot \sum_{(u,v) \in E_H} (\mathbf{f}(u) - \mathbf{f}(v))^2}
	{\sum_{u,v \in V_H} (\mathbf{f}(u) - \mathbf{f}(v))^2 \deg_H(u) \deg_H(v)}
	\right\}
    \enspace,
\end{equation}
where $\lambda_2(H)$ denotes the second smallest eigenvalue of the normalized Laplacian of $H$. %the summation in the dominator is over all unordered pairs of $u,v \in V_H$, and
%the eg $\vol_H(S)$ of a set $S\subseteq V_H$ is the sum of degrees of vertices in $S$, that is, $\vol_H(S) := \sum_{v \in S} \deg_H (v)$.
Let us consider the induced subgraph $H := G[C]$ on $C$.
Since $\phi(H) \ge \phi_{\myin}$,
Cheeger's inequality
yields $\lambda_2 (H) \ge \frac{\phi_{\myin}^2}{2}$. Therefore, if we apply this bound to inequality (\ref{ineq:from-Chung}), then,
\begin{eqnarray*}
    \vol_H(V_H) \cdot \frac{2\cdot\sum_{(u,v) \in E_H}(\xx_i(u) - \xx_i(v))^2}
	{\sum_{u,v \in V_H} (\xx_i(u) - \xx_i(v))^2 \deg_H(u) \deg_H(v)}
	\ge
    \lambda_2(H)
	\ge
    \frac{\phi_{\myin}^2}{2}
    \enspace.
%\label{ineq:eigenlowerbound}
\end{eqnarray*}
Combining this with the fact that \[
    \sum_{(u,v) \in E_H}(\xx_i(u) - \xx_i(v))^2 \le \sum_{(u,v) \in E_G} (\xx_i(u) - \xx_i(v))^2 \le \lambda_k,
\] where the last inequality follows from inequality (\ref{eqn:lambdaphiout}), we have that
\begin{displaymath}
    \sum_{u,v \in V_H} (\xx_i(u) - \xx_i(v))^2 \deg_H(u) \deg_H(v)
	\le
    \frac{4\lambda_k \vol_H(V_H)}{\phi_{\myin}^2}
    \enspace.
\end{displaymath}
Next, since $\phi(H) \ge \phi_{\myin} > 0$ implies that $\deg_H(u)\geq 1$ for all $u \in V_H$. Using the bound above we obtain:
\begin{eqnarray*}
\sum_{u,v \in V_H} (\xx_i(u) - \xx_i(v))^2
	&\le&
    \sum_{u,v \in V_H} (\xx_i(u) - \xx_i(v))^2 \deg_H(u) \deg_H(v) \\
	&\le&
    \frac{4\lambda_k \vol_H(V_H)}{\phi_{\myin}^2}
    \le
    \frac{4\lambda_k \vol(C)}{\phi_{\myin}^2}.
\end{eqnarray*}
This completes the proof of the lemma.
\end{proof}

\section{Missing Proofs from Section~\ref{sec:clustering}}
\label{sec:facts}

\begin{proof}[Proof of Fact~\ref{fact:eigenvalue-A-B}]
Since $\max_{i\leq p}\norm{\mat{X}_{\row(i)}-\mat{Y}_{\row(i)}}_2\leq \delta$,
then the Frobenius norm $\norm{\mat{X}-\mat{Y}}_F$ of $\mat{X}-\mat{Y}$ is at
most $\sqrt{p}\cdot\delta$, and therefore, the induced 2-norm
$\norm{\mat{X}-\mat{Y}}_2$ of $\mat{X}-\mat{Y}$ is at most
$\sqrt{p}\cdot\delta$. By Weyl's inequality~\cite{Tao12:randommatrix}, for any
$i$, $|\eta_i(\mat{X})-\eta_i(\mat{Y})|\leq |\eta_1(\mat{X}-\mat{Y})| \leq
\norm{\mat{X}-\mat{Y}}_2\leq \sqrt{p}\cdot\delta$.
\end{proof}

\begin{proof}[Proof of Fact~\ref{fact:rowdifference}]
For simplicity, let $\mat{X}_i:=\mat{X}_{\row(i)}$ and
$\mat{Y}_i:=\mat{Y}_{\row(i)}$. Note that the $i,j$th entry of $\mat{X}\cdot
\mat{X}^T$ and $\mat{Y}\cdot \mat{Y}^T$ are $\langle \mat{X}_i,\mat{X}_j\rangle$
and $\langle \mat{Y}_i,\mat{Y}_j\rangle$, respectively, and that
\begin{eqnarray*}
\abs{ \langle \mat{X}_i,\mat{X}_j\rangle - \langle \mat{Y}_i,\mat{Y}_j\rangle}
&=&
\abs{ \langle \mat{X}_i,\mat{X}_j\rangle - \langle \mat{Y}_i - \mat{X}_i + \mat{X}_i, \mat{Y}_j - \mat{X}_j + \mat{X}_j\rangle}\\
%&=&
%\abs{ \langle X_i,X_j\rangle - \langle Y_i - X_i, Y_j - X_j\rangle - \langle Y_i - X_i, X_j\rangle - \langle X_i, Y_j - X_j\rangle - \langle X_i, X_j\rangle} \\
&\leq &
\abs{\langle \mat{Y}_i - \mat{X}_i, \mat{Y}_j - \mat{X}_j\rangle} + \abs{\langle \mat{Y}_i - \mat{X}_i, \mat{X}_j\rangle} + \abs{\langle \mat{X}_i, \mat{Y}_j - \mat{X}_j\rangle}\\
&\leq&
\norm{\mat{Y}_i - \mat{X}_i}_2 \norm{\mat{Y}_j - \mat{X}_j}_2 + \norm{\mat{Y}_i - \mat{X}_i}_2 \norm{\mat{X}_j}_2 + \norm{\mat{X}_i}_2 \norm{\mat{Y}_j - \mat{X}_j}_2\\
&\leq&
\delta^2 + 2\gamma \delta.
\end{eqnarray*}
Therefore,
$\max_{i\leq p}\norm{(\mat{X}\cdot \mat{X}^T)_{\row(i)}-(\mat{Y}\cdot \mat{Y}^T)_{\row(i)}}_2 \leq \sqrt{p}(\delta^2 + 2\gamma \delta)$. %, and $\max_{i\leq q}\norm{(X\cdot X^T)_{col(i)}-(Y\cdot Y^T)_{col(i)}}_2 \leq \sqrt{p}(\delta^2 + 2\gamma \delta)$.
\end{proof}

\section{Experimental Evaluation}\label{sec:experiments}
Results from our greedy $k$-clustering implementation are shown
in Figures~\ref{experiment1}, \ref{experiment2}, \ref{experiment3}.
Cluster assignments for graphs are shown as colored nodes.
In the case where the graph comes from a triangulated surface, we have extended
the coloring to a small surface patch in the vicinity of the node. Each
experiment includes a plot of the eigenvalues of the normalized Laplacian (y-axis), by eigenvector number (x-axis).
A small rectangle on each plot highlights the corresponding spectral gap
between $k$ and $k+1$.

\paragraph*{Multiple spectral gaps}
Recall that graphs
which have a sufficiently large spectral gap between the $k$-th and $(k+1)$-st
eigenvalues admit a strong clustering~\cite{DBLP:conf/soda/GharanT14}.
Figure~\ref{experiment1} shows the result of our algorithm on a graph
with two prominent spectral gaps, $k=2$ (left) and $k=5$ (right).

This graph is sampled from the following generative model.
Let $C_1,\ldots,C_5$ be disjoint vertex sets of equal size, depicted as circles.
Every edge with both endpoints in the same $C_i$ appears with probability $p_1$,
every edge between $C_1\cup C_2$ and $C_3\cup C_4\cup C_5$ appears with
probability $p_3$, and every other edge appears with probability $p_2$, for
some $p_1\gg p_2 \gg p_3$. The resulting graph admits a strong $k$-partition,
for any $k\in \{2,5\}$. This fact is reflected in the output of our algorithm.

We remark that to achieve a sufficiently large gap at $k=5$ many intra-circle
edges are necessary, which makes the resulting figures too dense.
To make the plots readable, we have displayed only a subsampling of these edges.

\begin{figure}
\begin{center}
\begin{tabular}{ll}
\includegraphics[width=0.47\textwidth]{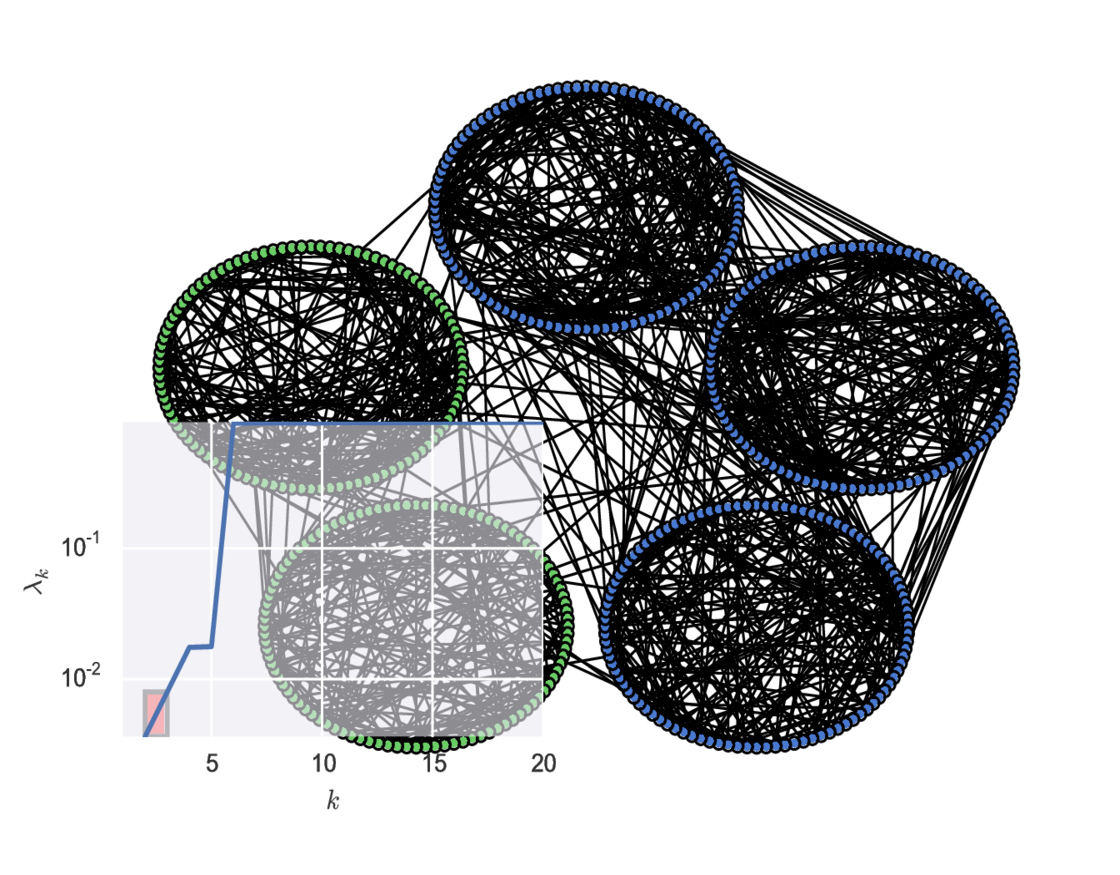} &
\includegraphics[width=0.47\textwidth]{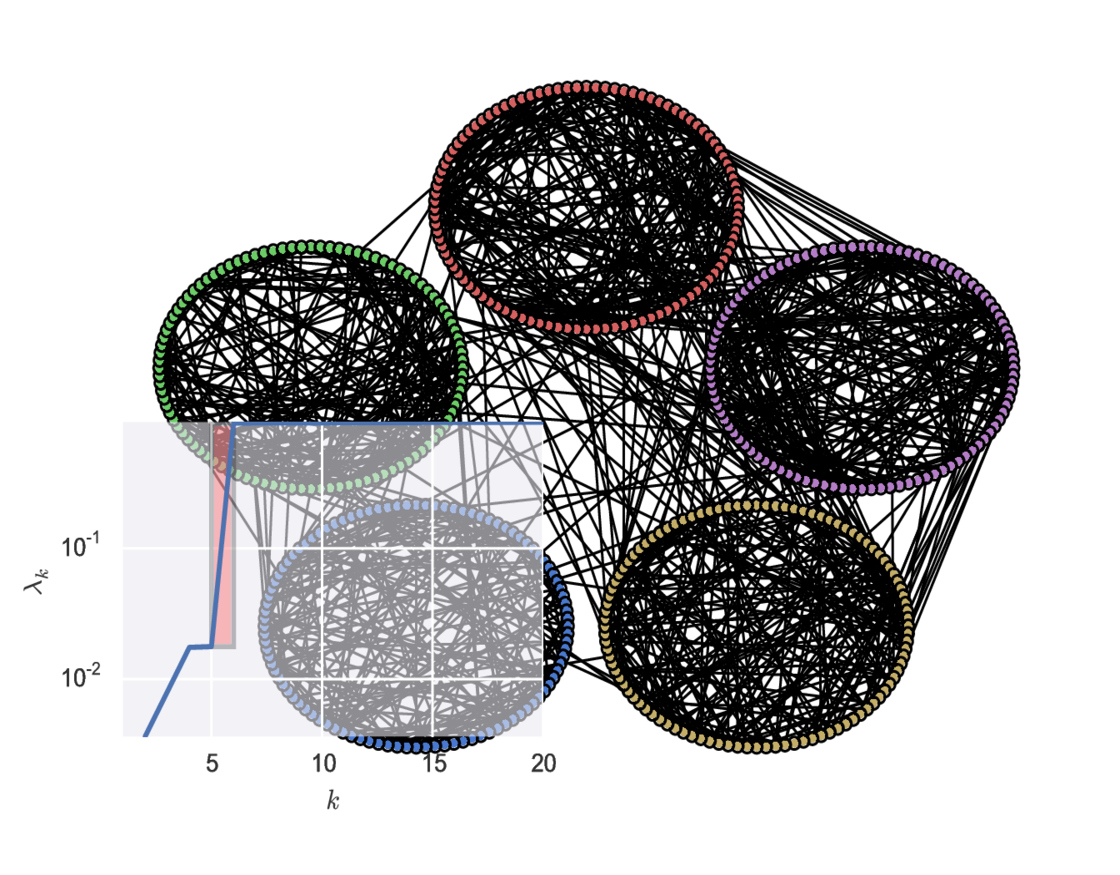}
\end{tabular}
\end{center}
\caption{Clustering a graph with multiple spectral gaps. The graph has $n=500$
nodes, with edges chosen according to our model below. Clustering is performed
with $R=18/(k\dmax\sqrt{n})$ for $k=2$ and $k=5$.}
\label{experiment1}
\end{figure}

\begin{figure}
\begin{center}
\begin{tabular}{ll}
\includegraphics[width=0.45\textwidth]{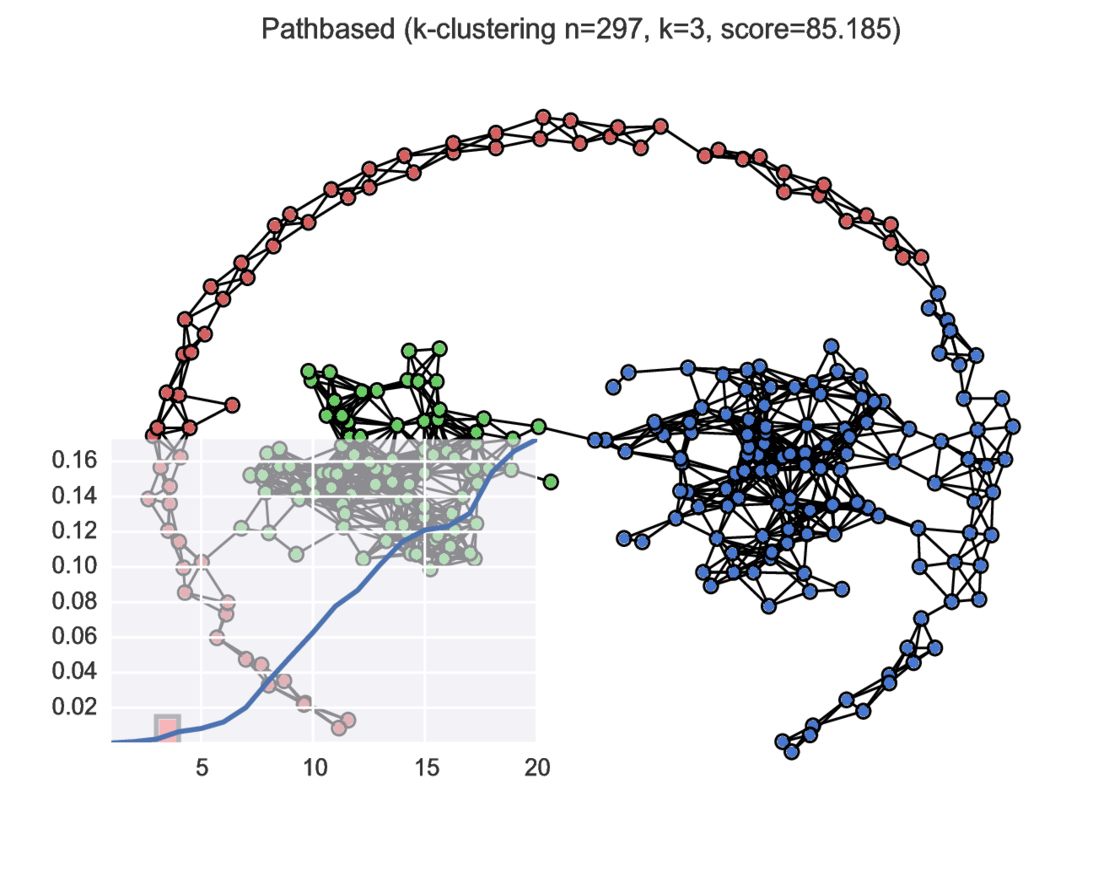} &
\includegraphics[width=0.45\textwidth]{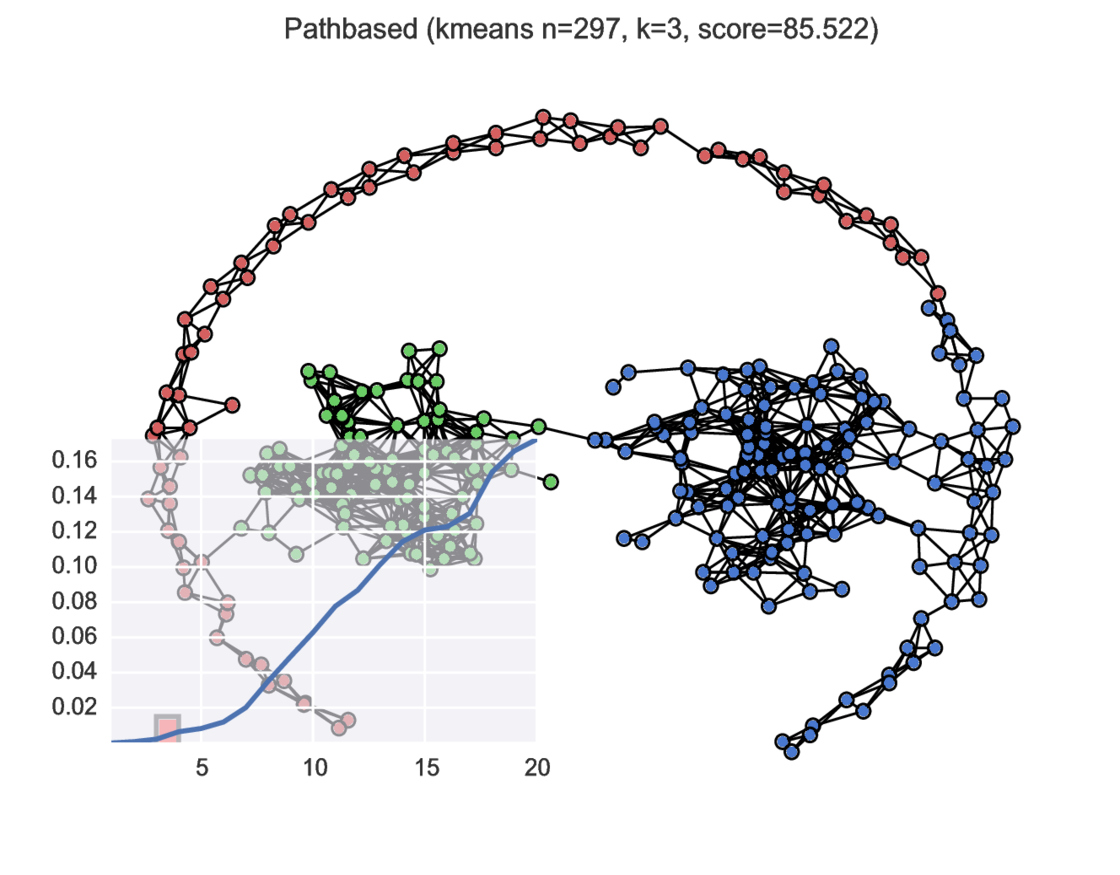} \\
\includegraphics[width=0.45\textwidth]{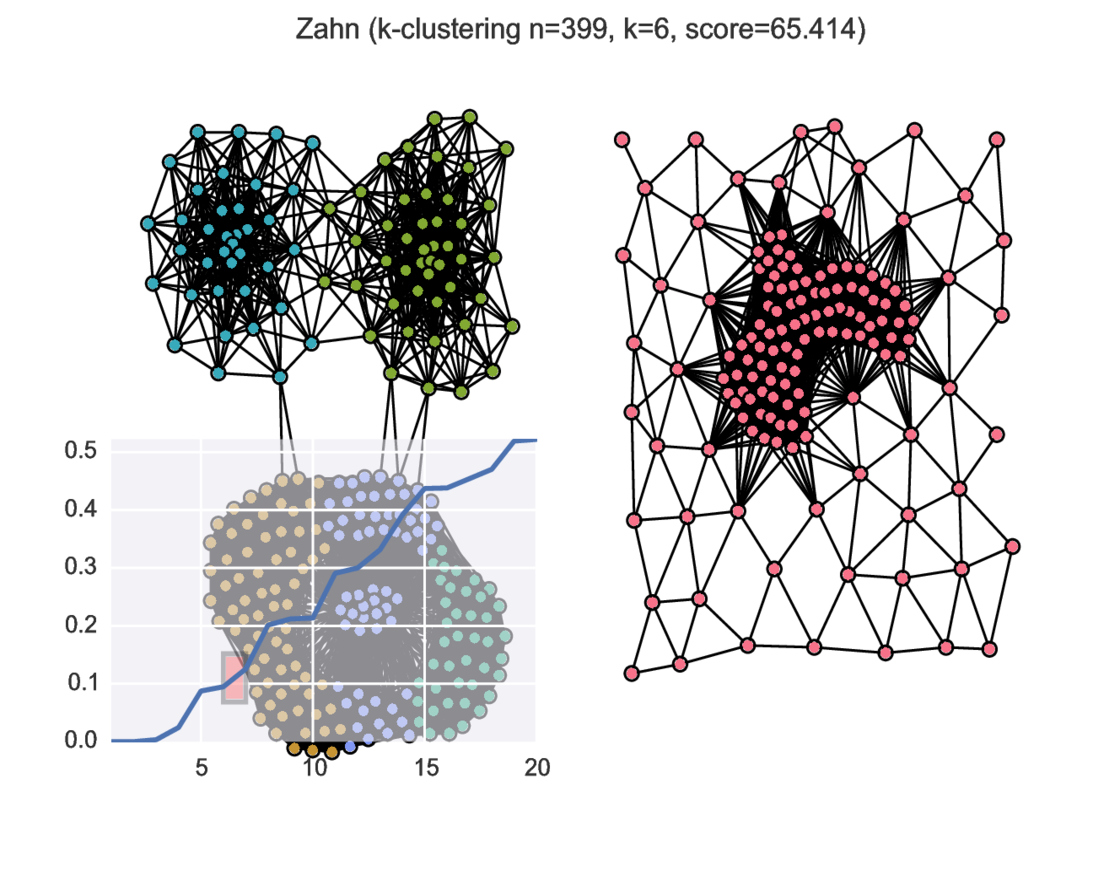} &
\includegraphics[width=0.45\textwidth]{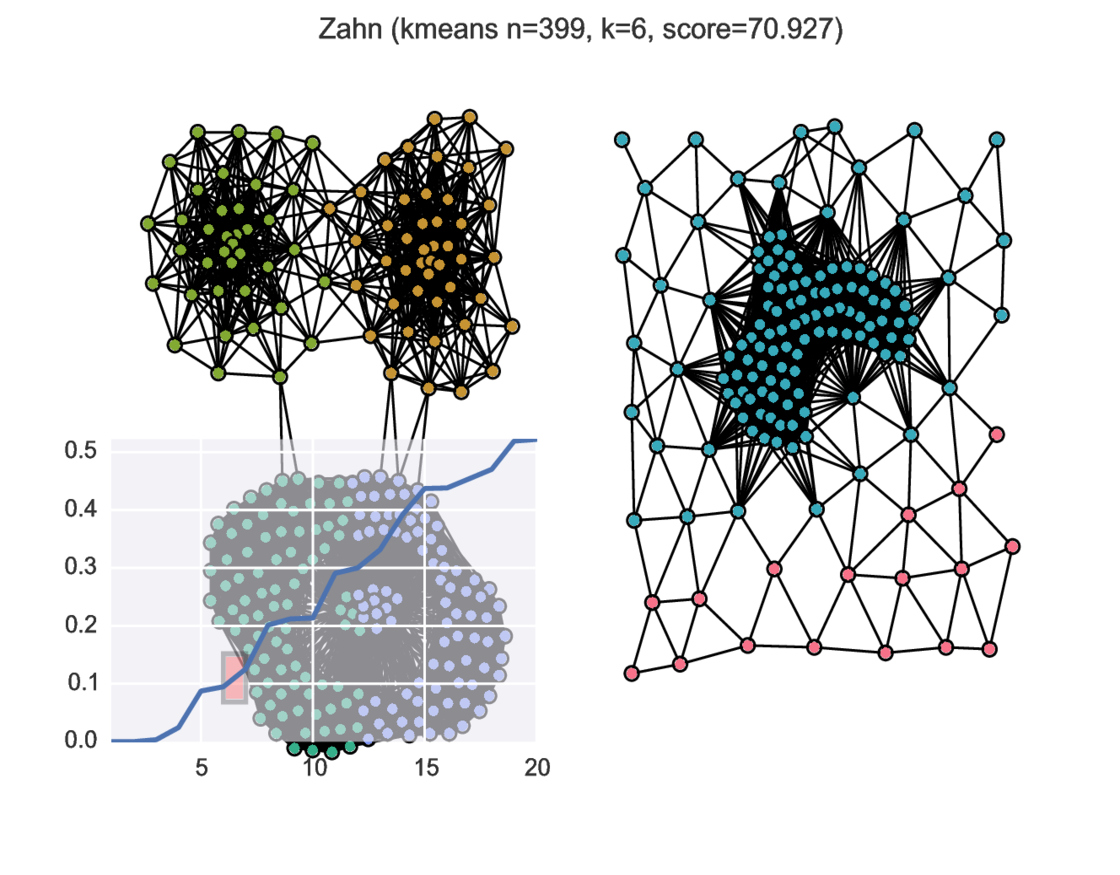} \\
\includegraphics[width=0.45\textwidth]{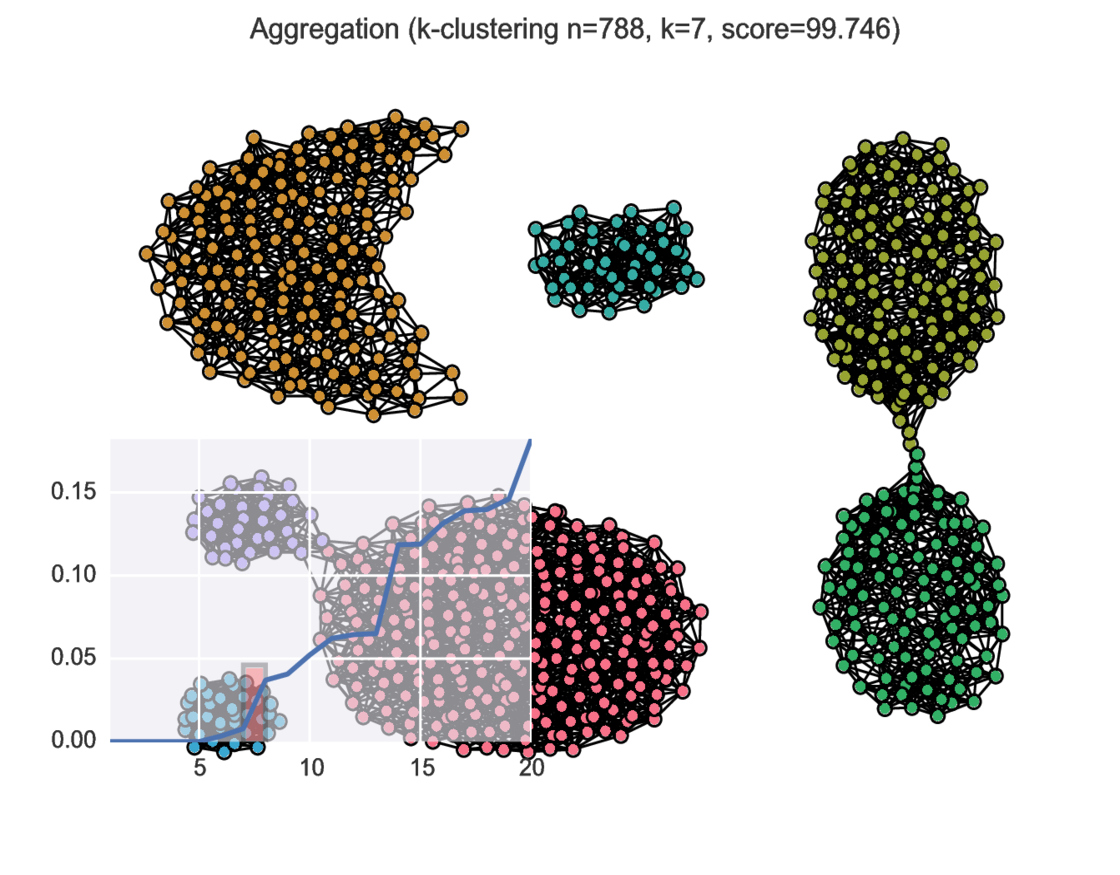} &
\includegraphics[width=0.45\textwidth]{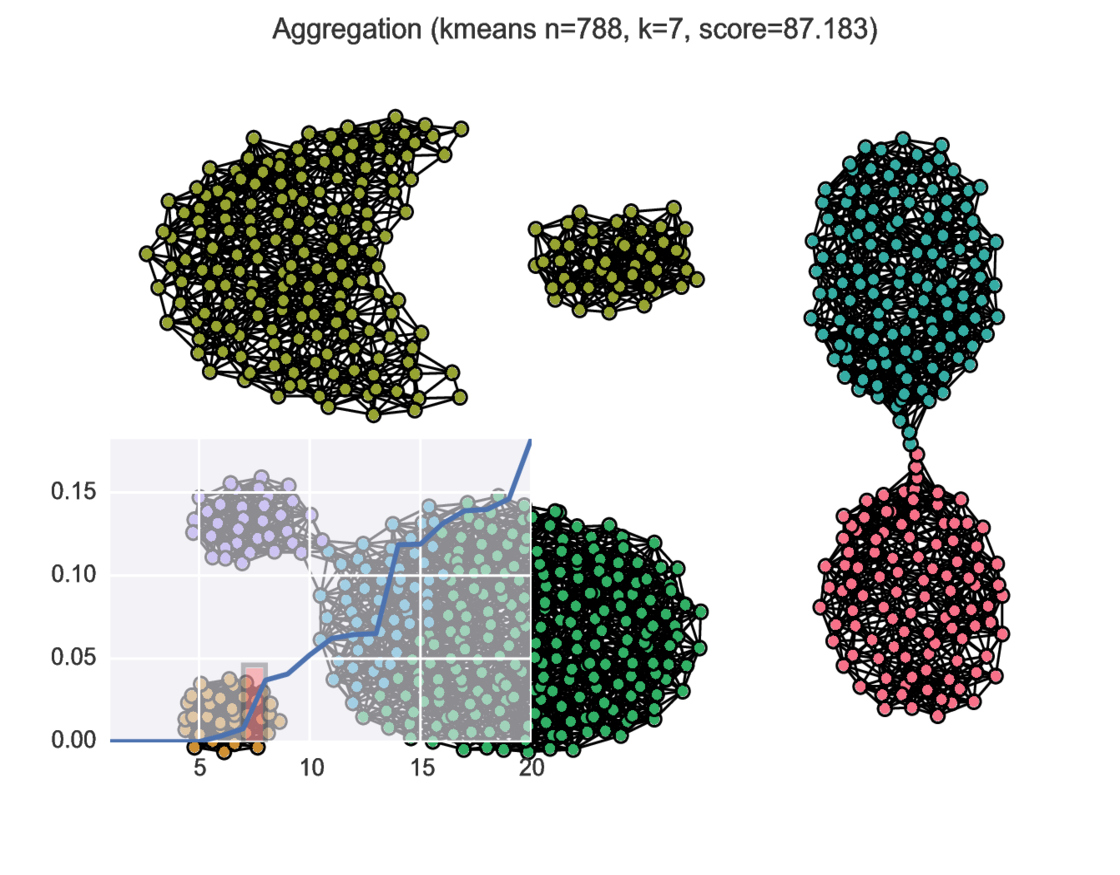} \\
\end{tabular}
\end{center}
\caption{Visual comparison of cluster assignments between Greedy $k$-Center
and $k$-means. Clustering was performed on graphs with radius $R=18/(k\dmax\sqrt{n})$.
The scores in the titles indicate percent correct classification with respect
to the ground truth.}
\label{experiment2}
\end{figure}

\paragraph*{Comparison with $k$-means}
In Figure~\ref{fig:expresults} we compare the greedy approach to $k$-means on
the spectral embedding. Our experiments build a graph on two and
three-dimensional euclidean point-cloud data by selecting a threshold value,
$D$, and connecting any two points which are no further than $D$. Additionally,
any singletons are removed. A key feature of the chosen point-cloud data is that
it comes with ground truth labeling which we lift to the corresponding
graph.

\begin{figure}
\begin{center}
\begin{tabular}{ l | c c c c c }
  Data set     &  n   & k & D   & Greedy $k$-C & $k$-means	\\
\hline
  pathbased   & 297  & 3 & 2   & 85.185     & 85.522        \\
  jain        & 373  & 2 & 3   & 100        & 100           \\
  zahn        & 399  & 6 & 3   & 65.414     & 71.471        \\
  aggregation & 788  & 7 & 2   & 99.746     & 96.063        \\
  LSun        & 400  & 3 & 1   & 93.250     & 93.250        \\
  Tetra       & 400  & 4 & 1   & 100        & 100           \\
  Hepta       & 212  & 7 & 5   & 44.811     & 57.849        \\
  Chainlink   & 1000 & 2 & 1   & 80.600     & 83.196        \\
  EngyTime\footnotemark
              & 4082 & 2 & 0.5 & 96.546     & 96.423        \\
  TwoDiamonds & 800  & 2 & 1   & 100        & 100           \\
\end{tabular}
\end{center}
\caption{Percentage agreement to ground truth for Greedy $k$-Clustering and
$k$-means. The results reported in the $k$-means column are the mean percent
agreement over $100$ trials where $k$-means was initialized by random selection
of $k$ points in the embedding. The $R$ parameter for Greedy $k$-Clustering was
given by the formula $R=18/(k\dmax\sqrt{n})$. The first four input
graphs are based on data sets of the same name from \cite{dataset:shapesets},
the rest are based on similarly named point-cloud data from FCPS.
\cite{dataset:FCPS}. \label{fig:expresults}}
\end{figure}
\footnotetext{Only the largest component was used for clustering.}

\begin{figure}
\begin{center}
\begin{tabular}{ll}
\includegraphics[width=0.4\textwidth]{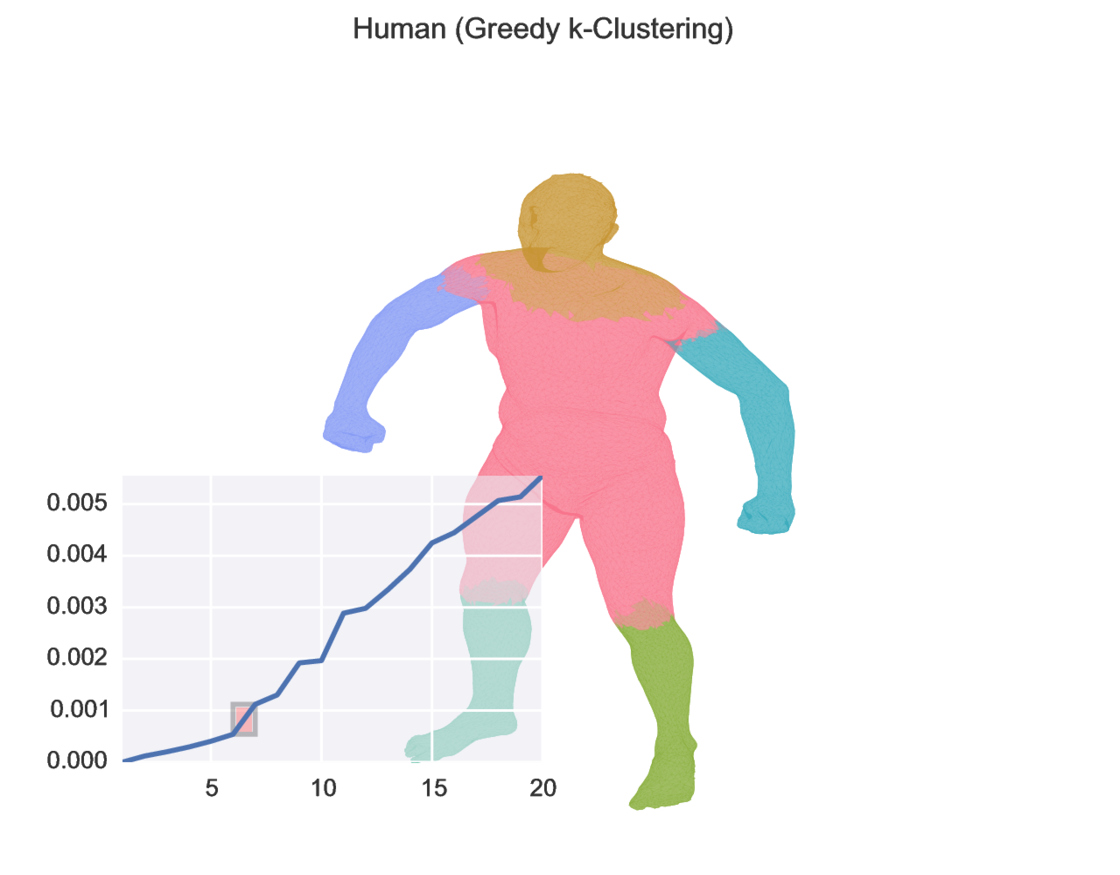} &
\includegraphics[width=0.4\textwidth]{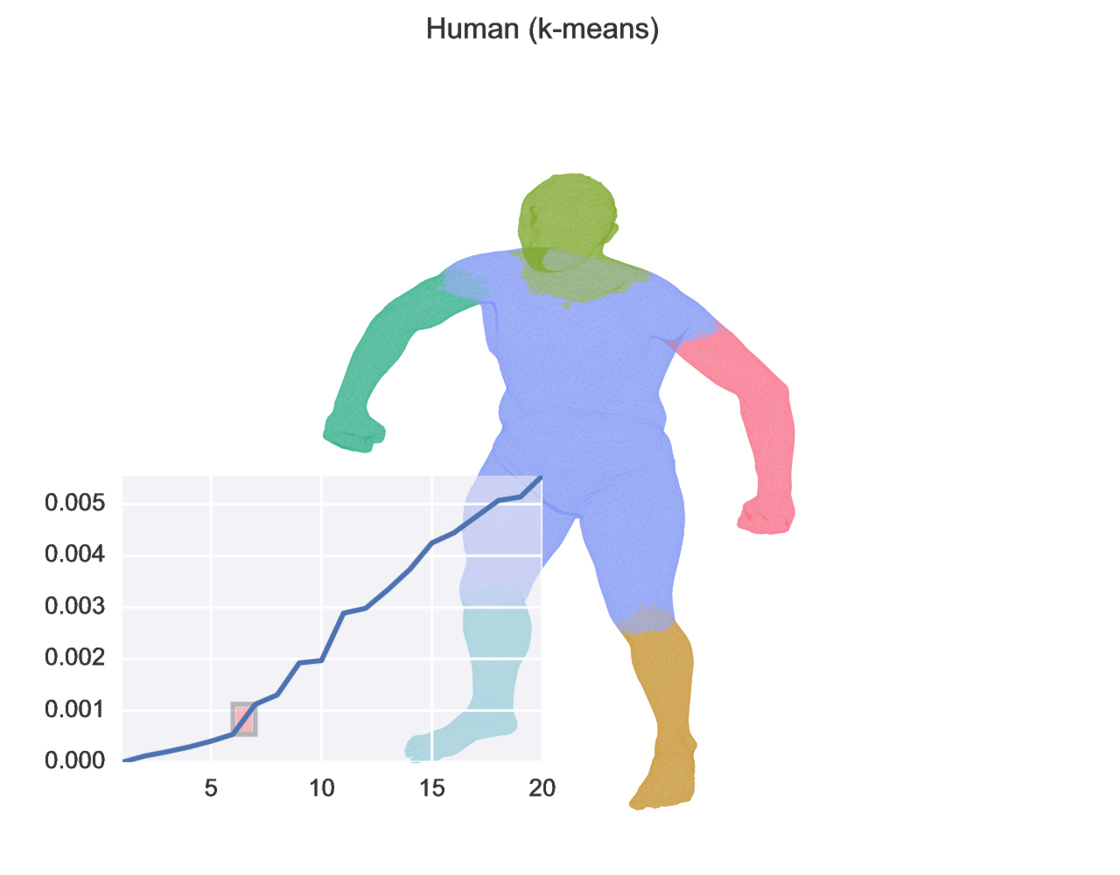} \\
\end{tabular}
\end{center}
\caption{A $k=6$ clustering performed on the $1$-skeleton of triangulated mesh
($n=12500$), using $R=30/(k\dmax\sqrt{n})$. Here the two algorithms agree
on $97.176\%$ of the vertices.}
\label{experiment3}
\end{figure}

%\section{Proof of Lemma~\ref{lem:clusterable-eigenvector}}
\end{document}